\documentclass[a4paper,11pt,reqno]{amsart}

\pdfoutput=1

\usepackage{a4wide,color,array}
\usepackage{cite}
\usepackage{hyperref}
\usepackage{amsmath,amssymb,amsthm,color}
\hypersetup{linktocpage,colorlinks, citecolor={magenta}}
\usepackage[english]{babel}
\usepackage[utf8]{inputenc}

\usepackage{xspace}
\usepackage{bm}				
\usepackage{bbm,upgreek}		
\usepackage[e]{esvect}		
\usepackage{esint}			
\usepackage{etoolbox}			
\usepackage{calc}				
\usepackage{eso-pic}			
\usepackage{datetime}			

\usepackage{lmodern}
\numberwithin{equation}{section}

\usepackage{enumitem}
\setlist{nosep}
\setlist{noitemsep}

\newcommand{\N}{\mathcal{N}}
\newcommand{\Z}{\mathbb{Z}}

\newcommand{\R}{\mathbb{R}}

\newcommand{\I}{\mathcal{I}}

\newtheorem{theorem}{Theorem}
\newtheorem{proposition}{Proposition}[section]
\newtheorem{lemma}[proposition]{Lemma}
\newtheorem{corollary}[proposition]{Corollary}
\newtheorem{remark}[proposition]{Remark}

\newtheorem*{notn}{Notation}

\theoremstyle{plain}

\theoremstyle{definition}

\newcommand{\tref}[1]{Theorem~\ref{t.#1}}
\newcommand{\pref}[1]{Proposition~\ref{p.#1}}
\newcommand{\lref}[1]{Lemma~\ref{l.#1}}
\newcommand{\cref}[1]{Corollary~\ref{c.#1}}

\newcommand{\sref}[1]{Section~\ref{s.#1}}
\newcommand{\rref}[1]{Remark~\ref{r.#1}}

\newcommand{\eref}[1]{(\ref{e.#1})}

\def \supp{\mathrm{supp }} 
\def \1{\mathbf{1}} 
\def \mcl{\mathcal}
\def \mbb{\mathbb}
\def \ep{\varepsilon}
\def \dist{\mathrm{dist}}
\newcommand{\g}{\mathsf{g}}

\def\nab{\nabla}

\def\({\left(}
\def\){\right)}

\def\XXint#1#2#3{{\setbox0=\hbox{$#1{#2#3}{\int}$}
		\vcenter{\hbox{$#2#3$}}\kern-.5\wd0}}

\newcommand{\fluct}{\mathrm{fluct}}
\newcommand{\Fluct}{\mathrm{Fluct}}
\newcommand{\Disc}{\mathrm{Disc}}

\newcommand{\Iso}{\mathrm{Iso}}

\def\nab{\nabla}
\def\pa{\partial}

\def\Ann{\mathrm{Ann}}

\newcommand{\M}{\mathcal{M}}

\newcommand{\E}{\mathbb{E}}
\renewcommand{\P}{\mathbb{P}}

\newcommand{\mueq}{\mu_{\mathrm{eq}}} 

\newcommand{\ov}{\overline}

\def\Leb{\mathrm{Leb}}

\makeatletter
\def\namedlabel#1#2{\begingroup
	#2%
	\def\@currentlabel{#2}%
	\phantomsection\label{#1}\endgroup
}
\makeatother

\def \diam{\mathrm{diam}}

\begin{document}
	\title[Overcrowding in Coulomb Gases]{Overcrowding and Separation Estimates for the Coulomb Gas}
	\author{Eric Thoma}
	\date{February 20, 2023}
	
	\begin{abstract}
	We prove several results for the Coulomb gas in any dimension $d \geq 2$ that follow from {\it isotropic averaging}, a transport method based on Newton's theorem. First, we prove a high-density Jancovici-Lebowitz-Manificat law, extending the microscopic density bounds of Armstrong and Serfaty and establishing strictly sub-Gaussian tails for charge excess in dimension $2$. The existence of microscopic limiting point processes is proved at the edge of the droplet. Next, we prove optimal upper bounds on the $k$-point correlation function for merging points, including a Wegner estimate for the Coulomb gas for $k=1$. We prove the tightness of the properly rescaled $k$th minimal particle gap, identifying the correct order in $d=2$ and a three term expansion in $d \geq 3$, as well as upper and lower tail estimates. In particular, we extend the two-dimensional ``perfect-freezing regime" identified by Ameur and Romero to higher dimensions. Finally, we give positive charge discrepancy bounds which are state of the art near the droplet boundary and prove incompressibility of Laughlin states in the fractional quantum Hall effect, starting at large microscopic scales. Using rigidity for fluctuations of smooth linear statistics, we show how to upgrade positive discrepancy bounds to estimates on the absolute discrepancy in certain regions. 
	\end{abstract}
	
	\maketitle
	
	\section{Introduction} \label{s.intro}
	\subsection{The setting.}\ \label{s.intro1}
	For $d \geq 2$, the $d$-dimensional Coulomb gas (or one-component plasma) at inverse temperature $\beta \in (0,\infty)$ is a probability measure on point configurations $X_N = (x_1,\ldots,x_N) \in (\R^d)^N$ given by
	\begin{equation} \label{e.P1def}
		\P^{W}_{N,\beta}(dX_N) = \frac{1}{\mcl Z} \exp\(-\beta \mcl H^{W}(X_N)\) dX_N
	\end{equation}
	where $dX_N$ is Lebesgue measure on $(\R^d)^N$, $\mcl Z$ is a normalizing constant, and
	\begin{equation} \label{e.mclHdef}
		\mcl H^{W}(X_N) = \frac12 \sum_{1 \leq i \ne j \leq N} \g(x_i - x_j) +  \sum_{i = 1}^N W(x_i)
	\end{equation}
	is the Coulomb energy of $X_N$ with confining potential $W$. The kernel $\g$ is the Coulomb interaction given by
	\begin{equation}
		\g(x) = \begin{cases}
			-\log |x| \quad &\text{if } d = 2 \\
			\frac{1}{|x|^{d-2}} \quad &\text{if } d \geq 3.
		\end{cases}
	\end{equation}
	While we will rarely require it, we have in mind the scaling $W = V_N := N^{2/d}V(N^{-1/d} \cdot)$ for a potential $V$ satisfying certain conditions. The normalization in $V_N$ is chosen so that the typical interstitial distance is of size $O(1)$, i.e.\ the Coulomb gas $\P^{V_N}_{N,\beta}$ is on the ``blown-up" scale. However, unless otherwise stated, we will work only under the assumption that $\Delta W$ exists and is bounded from above on $\R^d$ and \eref{P1def} is well-defined, though see \rref{A1} for comments on how this can be loosened significantly. For some results, we will need additional assumptions on $W$.
	
	Up to normalization, the kernel $\g$ gives the repulsive interaction between two positive point charges, and so the Coulomb gas exhibits a competition between particle repulsion, given by the first sum in \eref{mclHdef}, and particle confinement, given by the second sum in \eref{mclHdef}. The behavior of $X_N$ at the macroscopic scale (i.e.\ in a box of side length $O(N^{1/d})$) is largely dictated by the equilibrium measure $\mueq$, a compactly supported probability measure on $\R^d$ solving a variational problem involving $V$, in the sense that the empirical measure $N^{-1}\sum_{i=1}^N \delta_{N^{-1/d}x_i}$ is well-approximated weakly by $\mueq$ with high probability as $N \to \infty$. In particular, the rescaled points condense on the {\it droplet}, i.e.\ the support of $\mueq$. Even on mesoscales $O(N^\alpha)$, $0 < \alpha < 1/d$, the equilibrium measure gives a good approximation for particle density. Letting $\mueq^N$ be defined by $\mueq^N(A) = N\mueq(N^{-\frac1d}A)$ for measurable $A \subset \R^d$, one can form the random fluctuation measure
	\begin{equation} \label{e.fluctdef}
		\fluct(dx) = \sum_{i=1}^N \delta_{x_i}(dx) - \mueq^N(dx),
	\end{equation}
	which, despite being of size $O(N)$ in total variation, is typically of size $O(1)$ when acting on smooth functions (e.g.\ \cite{S22,LS18, AS21, BBNY19}).
	
	Most of the time, we will work with the more general probability measure $\P^{W,U}_{N,\beta}$ defined by
	\begin{equation}
		\P^{W,U}_{N,\beta}(dX_N) \propto \exp\(-\beta \mcl H^{W,U}(X_N)\) dX_N, \quad \mcl H^{W,U}(X_N) = \mcl H^W(X_N) + U(x_1,\ldots,x_N),
	\end{equation}
	where $U = U_N : (\R^d)^N \to \R$ is symmetric, superharmonic and locally integrable in each variable $x_i$, and such that the measure $\P^{W,U}_{N,\beta}$ is well-defined. Measures of this form capture behavior of the gas under conditioning. For example, the Coulomb gas \eref{P1def} disintegrates along $x_n$ to $\P^{W,U}_{N-1,\beta}$ with $U(X_{N-1}) = \sum_{i=1}^{N-1} \g(x_i - x_N)$. They also play an important role in the study of the fractional quantum Hall effect; see \sref{intro_discrepancy} for further discussion as well the surveys \cite{R22,R22_2}.
	
	We will apply a transport-type argument, called {\it isotropic averaging}, to give upper bounds for $\P^{W,U}_{N,\beta}$ on a variety of events, all concerning the overcrowding of particles. This terminology and a similar method was first used in \cite{L21}, but a technical issue limited its applicability. Our main contribution is to demonstrate that the method has wide-ranging applicability by giving relatively short and intuitive solutions to several open problems. We believe that it will be an important tool in future studies of the Coulomb gas.
	
	\subsection{A model computation.}\ \label{s.intro.model} The basic idea behind isotropic averaging will be motivated through the following model computation. We will refer to this computation, in more general forms, throughout the paper.
	
	We start by defining certain {\it isotropic averaging operators}. Given an index set $\I \subset \{1,2,\ldots,N\}$ and a rotationally symmetric probability measure $\nu$ on $\R^d$, we define
	$$
	\Iso_{\I, \nu} F((x_i)_{i \in \I})= \int_{(\R^d)^{\I}} F\( (x_i + y_i )_{i \in \I}\)  \prod_{i \in \I} \nu(dy_i)
	$$
	for any nice enough function $F : (\R^d)^{\I} \to \R$. We also consider the operator $\Iso_{\I,\nu}$ acting on functions of $X_N$, or more generally any set of labeled coordinates, by convolution with $\nu$ on the coordinates with labels in $\I$. For example, we have
	$$
	\Iso_{\I, \nu} F(X_N)= \int_{(\R^d)^{\I}} F\( X_N + (y_i \1_{\I}(i) )_{i =1}^N\)  \prod_{i \in \I} \nu(dy_i)
	$$
	by convention, and $\Iso_{\I,\nu} F(x_1) = F(x_1)$ if $1 \not \in \I$, otherwise $\Iso_{\I,\nu} F(x_1) = F \ast \nu(x_1)$.
	
	An important observation is that the kernel $\g$ is superharmonic everywhere and harmonic away from the origin, and thus we have the mean value inequality
	\begin{equation}\label{e.mvi}
		\Iso_{\I, \nu} \g(x_i - x_j) \leq \g(x_i - x_j) \quad \forall i,j.
	\end{equation}
	In our physical context, the isotropic averaging operator replaces each point charge $x_i$, $i \in \I$, by a charge distribution shaped like $\nu$ centered at $x_i$. Newton's theorem implies that the electric interaction between two disjoint, radial, unit charge distributions is the same as the interaction between two point charges located at the respective centers. More generally, if the charge distributions are not disjoint, then the interaction is more mild than that of the point charge system (this is because $\g(r)$ is decreasing in $r$ and the electric field generated by a uniform spherical charge is $0$ in the interior).
	
	Consider an event $E$ which we wish to show to be unlikely. For definiteness, we let $E$ be the event ``$B_r(z)$ contains at least $2$ particles" for some fixed $r \ll 1$ and $z \in \R^d$. By a union bound, we have
	\begin{equation} \label{e.model0}
		\P^{W,U}_{N,\beta}(E) \leq \sum_{i < j} \P^{W,U}_{N,\beta}(E_{\{i,j\}}) = \binom{N}{2} \P^{W,U}_{N,\beta}(E_{\{1,2\}}), \quad E_{\{i,j\}} := \{ x_i \in B_r(z)\} \cap \{x_j \in B_r(z)\}.
	\end{equation}
	We can bound the likelihood of $E_{\{1,2\}}$ by comparing each $X_N \in E_{\{1,2\}}$ to the weighted family of configurations generated by replacing $x_1$ and $x_2$ by unit charged annuli of inner radius $1/2$ and outer radius $1$. Letting $\nu$ be the uniform probability measure supported on the centered annulus $\Ann_{[1/2,1]}(0) \subset \R^d$, we have by Jensen's inequality
	\begin{align} \label{e.model1}
		\P^{W,U}_{N,\beta}(E_{\{1,2\}}) &= \frac{1}{\mcl Z}\int_{E_{\{1,2\}}} e^{-\beta \mcl H^{W,U}(X_N)} dX_N \leq \frac{e^{-\beta \Delta}}{\mcl Z} \int_{E_{\{1,2\}}} e^{-\beta \Iso_{\{1,2\},\nu} \mcl H^{W,U}(X_N)} dX_N \\ \notag &\leq \frac{e^{-\beta \Delta}}{\mcl Z} \int_{E_{\{1,2\}}} \Iso_{\{1,2\},\nu} e^{-\beta \mcl H^{W,U}(X_N)} dX_N
	\end{align}
	for
	$$
	\Delta = \inf_{X_N \in E_{\{1,2\}}} \mcl H^{W,U}(X_N) - \Iso_{\{1,2\},\nu} \mcl H^{W,U}(X_N).
	$$
	We can then consider the $L^2((\R^{d})^{\{1,2\}})$-adjoint of $\Iso_{\{1,2\},\nu}$, which we call $\Iso^\ast_{\{1,2\},\nu}$, to bound
	\begin{align} \label{e.model2}
		\P^{W,U}_{N,\beta}(E_{\{1,2\}}) &\leq \frac{e^{-\beta \Delta}}{\mcl Z} \int_{(\R^d)^N} \Iso^\ast_{\{1,2\},\nu} \1_{E_{\{1,2\}}}(X_N) e^{-\beta \mcl H^{W,U}(X_N)}dX_N \\\notag &= e^{-\beta \Delta} \E^{W,U}_{N,\beta}[\Iso^\ast_{\{1,2\},\nu} \1_{E_{\{1,2\}}}].
	\end{align}
	We call the above calculation, namely \eref{model1} and \eref{model2}, the {\it model computation}.
	There are now two tasks: (1) give a lower bound for $\Delta$ and (2) give an upper bound for the expectation of $\Iso^\ast_{\{1,2\},\nu} \1_{E_{\{1,2\}}}$.
	
	Regarding task (1), we expect $\Delta$ will be large: two particles initially clustered in $B_r(z)$ are replaced by annular charges of microscopic length scale, which interact mildly. It is a simple calculation to see the pairwise interaction between the charged annuli is bounded by $\g(1/2)$ (with the abuse of notation $\g(x) = \g(|x|)$). Regarding the potential term $\sum_{i=1}^N W(x_i)$ within $\mcl H^{W,U}(X_N)$, it will increase by at most a constant after isotropic averaging since $\Delta W \leq C$. The superharmonic term $U$ does not increase. Therefore, we have 
	$
	\Delta \geq \g(2r) - C.
	$
	
	Regarding task (2), since $\Iso_{\{1,2\}, \nu}$ is a convolution by $\nu^{\otimes 2}$, we have
	$$
	\Iso^\ast_{\{1,2\},\nu} \1_{E_{\{1,2\}}}(X_N) =\Iso_{\{1,2\},\nu} \1_{E_{\{1,2\}}}(X_N) \leq \| \nu \|^2_{L^\infty} \| \1_{E_{\{1,2\}}}(\cdot,\cdot,x_3,\ldots,x_N) \|_{L^1(\R^2)} \leq C r^{2d}.
	$$
	Moreover, we have $\Iso^\ast_{\{1,2\},\nu} \1_{E_{\{1,2\}}}(X_N) = 0$ if $x_1$ or $x_2$ is not in $B_{1+r}(z) \subset B_2(z)$. Thus
	$$
	\E^{W,U}_{N,\beta}[\Iso^\ast_{\{1,2\},\nu} \1_{E_{\{1,2\}}}] \leq Cr^{2d} \P^{W,U}_{N,\beta}(\{x_1,x_2 \in B_2(z)\}).
	$$
	Assembling the above, starting with \eref{model0}, we find
	\begin{equation} \label{e.model3}
		\P^{W,U}_{N,\beta}(E) \leq Ce^{-\beta \g(2r)}r^{2d} N^2 \P^{W,U}_{N,\beta}(\{x_1,x_2 \in B_2(z)\}).
	\end{equation}
	The probability appearing in the RHS will be bounded by $CN^{-2}$ by our microscopic local law \tref{1C.LL}, which is proved using a separate isotropic averaging argument, and we see that the probability of $E$ is bounded by $Cr^{2d} e^{-\beta \g(2r)}$. This is optimal in $d=2$, but can be improved to $Cr^{3d-2} e^{-\beta\g(2r)}$ in $d \geq 3$ (see \tref{1C.cluster}). The $CN^{-2}r^{2d}$ bound for the probability of $E_{\{1,2\}}$ comes from the decrease in phase space volume available to $x_1$ and $x_2$ from the full macroscopic scale of $O(N)$ volume per particle to a specific sub-microscopic ball of $O(r^d)$ volume upon restricting to $E_{\{1,2\}}$. In $d \geq 3$, the polynomial singularity of $\g$ generates additional effective constraints on $x_1$ and $x_2$ within $B_r(z)$.
	
	We remark that our technique exhibits perfect localization and gives quantitative estimates with computable constants. In particular, it is robust to certain types of conditioning and randomization of the ball $B_r(z)$, as well as allowing to prove disparate phenomena on vastly different scales. It can also be generalized to use operators other than $\Iso_{\I,\nu}$, as in the proof of \tref{etak.tight} where we give both upper and lower bounds on the minimal inter-particle difference. For the lower bound, we must apply our model computation with a ``mimicry" operator defined in \pref{createpairs}. The method, in particular techniques for estimating $\Delta$, can be made very precise, as in \tref{fLL.over}. Our model computation bears resemblance to the  Mermin-Wagner argument from statistical physics \cite{MW66}. It is also similar to an argument of Lieb, which applies only to ground states ($\beta = \infty$) and was generously shared and eventually generalized and published in \cite{NS15, RS16,  PS17}.

	\begin{notn}
		We identify $\P^{W,U}_{N,\beta}$ with the law of a point process $X$, with the translation between $X_N$ and $X$ given by $X = \sum_{i=1}^N \delta_{x_i}$. All point processes will be assumed to be simple. We also define the ``index" process $\mbb X$ given by $\mbb X(A) = \{i : x_i \in A\}$ for measurable sets $A$. For example, we have $E = \{X(B_r(z)) \geq 2\}$ and $E_{\{1,2\}} = \{\{1,2\} \subset \mbb X(B_r(z))\}$ for the events $E$ and $E_{\{1,2\}}$ considered in this subsection.
	\end{notn}
	
	\subsection{JLM laws.}\ Introduced in \cite{JLM93}, Jancovici-Lebowitz-Manificat (JLM) laws give the probability of large charge discrepencies in the Coulomb gas. The authors considered an infinite volume {\it jellium} and approximated the probability of an absolute net charge of size much larger than $R^{(d-1)/2}$ in a ball of radius $R$ as $R \to \infty$. The jellium is a Coulomb gas with a uniform negative background charge, making the whole system net neutral in an appropriate sense. Since the typical net charge in $B_R(0)$ is expected to be of order $R^{(d-1)/2}$ (see \cite{MY80}), the JLM laws are moderate to large deviation results and exhibit tail probabilities with very strong decay in the charge excess. The arguments of \cite{JLM93} are based on electrostatic principles and consider several different regimes of the charge discrepancy size.
	
	We are interested in a rigorous proof of the high density versions of the JLM laws. These versions apply when $X(B_R(z))$ exceeds the expected charge $\mueq(B_R(z))$ by a large multiplicative factor $C$. They predict that
	\begin{equation} \label{e.JLM.pred}
		\P_{\mathrm{jell}}(\{X(B_R(z)) \geq Q\}) \sim \begin{cases} \exp\(-\frac{\beta}{4} Q^2 \log \frac{Q}{Q_0}\) \quad \text{if } d=2, \\
			\exp\(-\frac{\beta}{4R} Q^2\) \quad \text{if } d=3,
		\end{cases}
	\end{equation}
	for $Q_0 = |B_R(z)|$. The prediction applies for $Q \gg R^d$ as $R \to \infty$
	
	Our main results prove the high density JLM law upper bounds in all dimensions in the ultra-high positive charge excess regime. We do so for $\P^{W,U}_{N,\beta}$, a Coulomb gas with potential confinement and superharmonic perturbation, though the result holds also for the jellium {\it mutatis mutandis}. We note that our result does not require $R \to \infty$. Indeed, we have found it very useful at $R=1$ as a local law upper bound valid on all of $\R^d$, an extension of the microscale local law in \cite{AS21} which is only proved for $z$ sufficiently far into the interior of the droplet and under other more restrictive assumptions. Note that although we do not obtain a sharp coefficient on $Q^2$ in the exponent of the $d \geq 3$ case, it could be improved with additional effort in \pref{1C.LL.iso}.
	
	\begin{theorem}[High Density JLM Law] \label{t.1C.LL}
		For any $R \geq 1$, integer $\lambda \geq 100$, and integer $Q$ satisfying
		\begin{equation} \label{e.K.cond.1}
			Q \geq \begin{cases}
				\frac{C \lambda^2 R^2 + C\beta^{-1} }{\log(\frac14 \lambda)} \quad &\text{if } d=2,\\
				C  R^{d} +C \beta^{-1} R^{d-2} \quad &\text{if } d \geq 3,
			\end{cases}
		\end{equation}
		we have
		\begin{equation}
			\P^{W,U}_{N,\beta}(\{X(B_R(z)) \geq Q \}) \leq \begin{cases}
				e^{-\frac12 \beta \log(\frac14 \lambda) Q^2 + C(1+\beta \lambda^2 R^2) Q} \quad &\text{if } d=2,\\
				e^{-2^{-d} \beta R^{-d + 2} Q(Q-1)} \quad &\text{if } d \geq 3,
			\end{cases}
		\end{equation}
		and the result remains true if $z$ is replaced by $x_1$. The constant $C$ depends only on the dimension and the upper bound for $\Delta W$. In particular if $d=2$ and $Q \geq C_{\beta,W} R^2$, we may choose $\lambda = \sqrt{\frac{Q}{R^2}}$ to see
		$$
		\P^{W}_{N,\beta}(\{X(B_R(z)) \geq Q \}) \leq e^{-\frac{\beta}{4} \log \(\frac{Q}{R^2}\) Q^2 + C\beta Q^2 + CQ}.
		$$
	\end{theorem}
	
	\begin{remark}
		The physical principles leading to the law \eref{JLM.pred} focus on the change of free energy between an unconstrained Coulomb gas and one constrained to have charge $Q$ in $B_R(z)$. For the constrained gas, the most likely particle configurations involve a build up of positive charge on an inner boundary layer of $B_R(z)$ and a near vacuum outside of $B_R(z)$ which ``screens" the excess charge. Since the negative charge density is bounded (in a jellium by definition and in $\P^{W,U}_{N,\beta}$ by $\Delta W \leq C$), the negative screening region must be extremely thick when $Q \gg R^d$. The self-energy of the negative screening region is the dominant contributor to the \eref{JLM.pred} bounds in \cite{JLM93}. In our proof, we apply an isotropic averaging operator that moves the particles within $B_R(z)$ to the bulk of the vacuum region, extracting a large average energy change per particle, thus providing a different perspective on the JLM law.
	\end{remark}
	
	\begin{remark}
		\tref{1C.LL} applies to small $\beta > 0$. In particular, one sees that charge excesses of order $T R^d$, $T \gg 1$, become unlikely as soon as $R \geq C^{-1}\beta^{-1/2}$. For this particular estimate type, \tref{1C.LL} therefore improves the minimal effective distance given in \cite[Theorem 1]{AS21} in dimensions $d = 2$ and $d \geq 5$ ($R \geq C\beta^{-1/2} (\log \beta^{-1})^{1/2}$ and $R \geq C \beta^{\frac{1}{d-2} - 1}$, respectively).
	\end{remark}
	
	\tref{1C.LL} immediately allows us to generalize \cite[Corollary 1.1]{AS21}, which established the existence of limiting microscopic point processes for $(x_1,\ldots,x_N)$ re-centered around a point $z$. Previous to the work of Armstrong and Serfaty, the existence of such a process was only known in $d=2$ and $\beta=2$, where it is the Ginibre point process with an explicit correlation kernel. In \cite{AS21}, the point $z$ must be in the droplet $\Sigma_N$ and a mesoscopic distance $CN^{\frac{1}{d+2}}$ distance from the edge of the droplet $\pa \Sigma_N$. We are able to lift that restriction, and in particular we can take $W = V_N$ and $z = z_N$ near or in $\pa \Sigma_N$, in which case one would expect a genuinely different limit than the bulk case.
	\begin{corollary} \label{c.point.process}
		For any sequence of points $z_N \in \R^d$, the law under $\P^{W,U}_{N,\beta}$ of the point process $\sum_{i=1}^N \delta_{x_i - z_N}$ converges weakly along subsequences as $N \to \infty$ to a simple point process.
	\end{corollary}
	\begin{proof}
		Tightness of the law of the finite dimensional distributions $(X(A_1),\ldots,X(A_n))$ for bounded Borel sets $A_1,\ldots,A_n$ (or for shifted versions of $X$) follows from \tref{1C.LL}. This implies weak convergence of the laws of the point processes (see \cite[Theorem 11.1.VII]{DV08}).
	\end{proof}
	\begin{remark}
		Any limit from \cref{point.process} will also enjoy analogs of \tref{1C.LL}, \tref{kpoint}, and \tref{fLL.over}.
	\end{remark}
	\subsection{Clustering and the $k$-point function.}\
	We have already seen in \sref{intro.model} that isotropic averaging can be applied to the description of the gas below the microscale, and we now state our full results. We are interested in pointwise bounds for the {\it $k$-point correlation function} $\rho_k$, defined by
	\begin{equation} \label{e.kpoint.def}
		\int_{A_1 \times A_2 \times \cdots \times A_k} \rho_k(y_1,y_2,\ldots,y_k) dy_1 \cdots dy_k = \frac{N!}{(N-k)!} \P^{W,U}_{N,\beta}\(\bigcap_{i=1}^k \{x_i \in A_i \}\)
	\end{equation}
	for measurable sets $A_1,A_2,\ldots,A_k \subset \R^d$.
	
	The functions $\rho_k$, and their truncated versions, are objects of intense interest. For example, in the physics literature, they are known to capture the charge screening behavior of the gas and satisfy sum rules and BBGKY equations \cite{GLM80,M88}. For $d=2$, spatial oscillations of $\rho_1$ are expected to occur for large enough $\beta$ \cite{CSA20, C06, CW03}. Starting at $\beta > 2$, the oscillations occur near the edge of the droplet, and as $\beta$ increases the oscillations penetrate the bulk of the droplet (numerically, it is present at $\beta = 200$) \cite{CSA20}. This phenomenon is part of a debated freezing or crystallization transition in the two-dimensional Coulomb gas \cite{KK16}.
	
	Many results on $\rho_k$ are known when integrated on the microscale or higher, though these results are often stated in terms of integration of the empirical measure $N^{-1}X$ against test functions. We will not comprehensively review previous results, but only mention that \cite{AS21} proves that $\int_{B_1(z)} \rho_1(y) dy$ is uniformly bounded in $N$ for $z$ sufficiently far inside the droplet.
	
	We will be interested in pointwise bounds on $\rho_k(y_1,\ldots,y_k)$, particularly when some of the $y_i$ within sub-microscopic distances of each other. One should see $\rho_k \to 0$ as $y_1 \to y_2$ due to the repulsion between particles. There are no previously existing rigorous results for pointwise values for general $\beta$; even boundedness of $\rho_1$ was until now unproved.
	
	\begin{theorem} \label{t.kpoint}
		We have that
		\begin{equation}
			\rho_1(y) \leq C
		\end{equation}
		for some constant $C$ independent of $N$ and $y$. We also have
		\begin{equation} 
			\rho_k(y_1,y_2,\ldots,y_k) \leq C \prod_{1 \leq i < j \leq k} (1 \wedge |y_i - y_j|^{\beta}) \quad \text{if } d=2
		\end{equation}
		and
		\begin{equation}
			\rho_k(y_1,y_2,\ldots, y_k) \leq C \exp\(-\beta \mcl H^0(y_1,\ldots,y_k) \) \quad \text{if } d\geq 3.
		\end{equation}
	\end{theorem}
	
	The following bound on sub-microscopic particle clusters is easily derived by integrating \tref{kpoint}. We point out the enhanced $r^{2d-2}$ factor in \eref{cluster.centered.3.Q2}, which will be crucial for \tref{etak.tight}.
	
	\begin{theorem} \label{t.1C.cluster}
		Let $Q$ be a positive integer. There exists a constant $C$, dependent only on $\beta$, $Q$, and $\sup \Delta W$, such that for all $r > 0$ we have
		\begin{equation} \label{e.cluster.fixed.2}
			\P^{W,U}_{N,\beta}(\{X(B_r(z)) \geq Q \}) \leq C r^{2Q + \beta \binom{Q}{2}} \quad \text{if } d=2,
		\end{equation}
		and
		\begin{equation} \label{e.cluster.fixed.3}
			\P^{W,U}_{N,\beta}(\{X(B_r(z)) \geq Q \}) \leq Cr^{dQ} e^{-\frac{\beta}{2^{d-2}} \cdot \frac{1}{r^{d-2}}\binom{Q}{2}}\quad \text{if } d \geq 3.
		\end{equation}
		We also have for $Q \geq 2$ and $d = 2$ that
		\begin{equation} \label{e.cluster.centered.2}
			\P^{W,U}_{N,\beta}(\{X(B_r(x_1)) \geq Q \}) \leq C r^{2(Q-1) + \beta \binom{Q}{2}}
		\end{equation}
		and $d \geq 3$ that
		\begin{equation} \label{e.cluster.centered.3}
			\P^{W,U}_{N,\beta}(\{X(B_r(x_1)) \geq Q \})  \leq Cr^{d(Q-1)} e^{-\frac{\beta}{2^{d-2}} \cdot \frac{1}{r^{d-2}}\binom{Q-1}{2}} e^{-\beta \frac{1}{r^{d-2}} (Q-1)}.
		\end{equation}
		In the case of $Q=2$ and $d\geq 3$ this can be improved to
		\begin{equation} \label{e.cluster.centered.3.Q2}
			\P^{W,U}_{N,\beta}(\{X(B_r(x_1)) \geq 2 \})  \leq Cr^{2d-2} e^{-\frac{\beta}{r^{d-2}}}.
		\end{equation}
	\end{theorem}
	
	We remark that the $k=1$ and $Q=1$ cases of \tref{kpoint} and \tref{1C.cluster}, respectively, are instances of {\it Wegner estimates}. In the context of $\beta$-ensembles on the line, Wegner estimates were proved in \cite{BMP21} and for Wigner matrices in \cite{ELS10}. The $Q=2$ cases of \tref{1C.cluster} are {\it particle repulsion estimates}. These estimates, as well as eigenvalue minimal gaps, have been considered for random matrices in many articles, e.g. \cite{NTV17,T13,TV11,EKYY12}.
	
	\begin{remark} We claim that our results in \tref{kpoint} are essentially optimal and that \tref{1C.cluster} is optimal if $d=2$ or $Q=2$. For $d,Q \geq 3$, one can improve \tref{1C.cluster} by more carefully integrating \tref{kpoint}, though an optimal, explicit solution for all $Q$ would be difficult. Our claim is evidenced by the tightness results we prove in \tref{etak.tight} and by computations for merging points with fixed $N$.
	\end{remark}
	\subsection{Minimal particle gaps.}\
	We will also study the law of the $k$th smallest particle gap $\eta_{k}$, i.e.\
	\begin{equation} \label{e.etak.def}
		\eta_k(X_N) = \text{the}\ k\text{th smallest element of}\ \{|x_i - x_j| \ : \ i,j \in \{1,\ldots,N\}, i \ne j\}.
	\end{equation}
	Note that the particle gaps $|x_i - x_j|$, $i \ne j$, are almost surely unique under $\P^{W,U}_{N,\beta}$.
	
	Previously, the order of $\eta_1$ was investigated dimension two in \cite{A18} and \cite{AR22}. The latter article proves that $\eta_1 \geq (N \log N)^{-\frac{1}{\beta}}$ with high probability as $N \to \infty$ for all $\beta > 1$. It is also proved that $\eta_1 > C^{-1}$ with high probability if $\beta = \beta_N$ grows like $\log N$. This suggests that the gas ``freezes", even at the level of the extremal statistic $\eta_1$, in this temperature regime.
	
	We remark that \tref{1C.cluster} already improves this bound.
	\begin{corollary} \label{c.eta1.bd.d2}
		Let $d=2$. We have
		\begin{equation} \label{e.eta1.bd.d2}
			\P^W_{N,\beta}(\{ \eta_1 \leq \gamma N^{-\frac{1}{2+\beta}} \}) \leq C \gamma^{2 + \beta} \quad \forall \gamma > 0
		\end{equation}
		for a constant $C$ independent of $N$.
	\end{corollary}
	\begin{proof}
		This follows from \tref{1C.cluster}, specifically \eref{cluster.centered.2} with $Q = 2$ and $r = \gamma N^{-\frac{1}{2+\beta}}$, and a union bound.
	\end{proof}
	Furthermore, an examination of the proof's dependence on $\beta$ shows we can let $\beta = \beta_N = c_0 \log N$ and take a constant $C$ equal to $e^{C\beta} = N^{C}$ on the RHS of \eref{eta1.bd.d2}. Letting $\gamma = N^{1/(2+\beta)} \gamma'$ for a small enough $\gamma' > 0$ in \cref{eta1.bd.d2} shows that $\eta_1$ is bounded below by a constant independent of $N$ with high probability, offering an alternative proof of the freezing regime identified in \cite{A18, AR22}. An identical idea using \eref{cluster.centered.3} in place of \eref{cluster.centered.2} proves that the gas freezes in dimension $d \geq 3$ in the $\log N$ inverse temperature regime as well.
	
	It is natural to wonder whether \cref{eta1.bd.d2} is sharp in some sense and whether there are versions for $\eta_k$ and $d \geq 3$. We will give an affirmative answer to these questions.
	
	For the eigenvalues of random matrices, the law of $\eta_k$ has been of great interest. While one dimensional (for the most studied cases), these models are particularly relevant to the $d=2$ case because the interaction between eigenvalues is also given by $\g = -\log$ for certain ensembles. In \cite{BAB13}, the authors prove for the CUE and GUE ensembles that $N^{4/3} \eta_1$ is tight and has limiting law with density proportional to $x^{3k-1} e^{-x^3}$. Note that the interstitial distance is order $N^{-1}$ for these ensembles, whereas for us it will be order $1$. The proof uses determinantal correlation kernel methods. Recently, the extremal statistics of generalized Hermitian Wigner matrices was studied by Bourgade in \cite{B22} with dynamical methods, proving that $\eta_1$ is also of order $N^{-4/3}$ and the rescaled limiting law is identical to the GUE case. For a symmetric Wigner matrix, \cite{B22} proves that the minimal gap is of order $N^{-3/2}$.
	
	The articles \cite{FW21, FTW19} prove even more detailed results on the minimal particle gaps for certain random matrix models. Specifically, for the GOE ensemble and the circular $\beta$ ensemble with positive integer $\beta$, the joint limiting law of the minimal particle gaps is determined. Actually, convergence of a point process containing the data of the minimal particle gaps and their location is proved, showing in particular that the gap locations are asymptotically Poissonian. The methods crucially rely on certain exact identities unavailable in our case.
	
	The only optimal $d=2$ result is in \cite{SJ12}. Here, the authors consider the $k$th smallest eigenvalue gap of certain normal random matrix ensembles. The Ginibre ensemble, after rescaling the eigenvalues by $N^{1/2}$, corresponds to $\P^{V_N}_{N,\beta}$ for a certain quadratic $V$ and $\beta = 2$. Using determinantal correlation kernel methods, \cite{SJ12} proves that, for the Ginibre ensemble, $N^{\frac{1}{4}} \eta_k$ is tight as $N \to \infty$ (with interstitial distance $O(1)$), and its limiting law on $\R$ has density proportional to $x^{4k-1} e^{-x^4} dx$. \tref{etak.tight} extends the tightness result to all $\beta$ and general potential $V$.
	
	For our main result, we will take $W = V_N$ and require some extra assumptions on $V$. These assumptions are only used to prove \lref{lonely}, which is used in proving lower bounds on $\eta_k$. In particular, our upper bounds on $\eta_k$ below hold for $\P^{W,U}_{N,\beta}$ in full generality. For \tref{etak.tight}, we require
	\begin{equation} \tag{A1} \label{e.A1}
		\lim_{|x| \to \infty} V(x) + \g(x) = +\infty \quad \text{and} \quad \int_{\R^d} e^{-\beta (V(x) - \log(1+|x|))} dx < \infty \quad \text{if }d=2,
	\end{equation}
	\begin{equation} \tag{A2} \label{e.A2}
		\exists \ep > 0 \quad \liminf_{|x| \to \infty} \frac{V(x)}{|x|^\ep} > 0 \quad \text{if } d \geq 3.
	\end{equation}

	\begin{theorem} \label{t.etak.tight}
		Let $W = V_N$ with $\Delta W \leq C$, in $d=2$ condition \eref{A1} satisfied, and in $d \geq 3$ condition \eref{A2} satisfied. Then, in $d=2$ the law of $N^{\frac{1}{2+\beta}} \eta_k$ is tight as $N \to \infty$. Moreover, we have
		$$
		\limsup_{N \to \infty} \P^{V_N}_{N,\beta}(\{N^{\frac{1}{2+\beta}} \eta_k \leq \gamma\}) \leq C \gamma^{k(2 + \beta)},
		$$
		$$
		\limsup_{N \to \infty} \P^{V_N}_{N,\beta}(\{N^{\frac{1}{2+\beta}} \eta_k \geq \gamma\}) \leq C\gamma^{-\frac{4+2\beta}{4+\beta}}
		$$
		for all $\gamma > 0$.
		
		In $d\geq 3$, let $Z_k$ be defined by
		$$
		\eta_k = \(\frac{\beta}{\log N} \)^{\frac{1}{d-2}} \( 1 + \frac{2d-2}{(d-2)^2}  \frac{ \log \log N}{\log N} + \frac{Z_k}{(d-2) \log N}  \).
		$$
		Then the law of $Z_k$ is tight as $N \to \infty$. We have
		$$
		\limsup_{N \to \infty} \P^{V_N}_{N,\beta}(\{Z_k \leq -\gamma\}) \leq C e^{-k \gamma},
		$$
		$$
		\limsup_{N \to \infty} \P^{V_N}_{N,\beta}(\{Z_k \geq \gamma\}) \leq C e^{- \frac12 \gamma}
		$$
		for $\gamma > 0$.
	\end{theorem}
	\begin{proof}
		The theorem follows easily from combining \pref{etak.isbig} and \pref{etak.issmall} from \sref{etak}.
	\end{proof}
	\begin{remark}
		The estimates for $\eta_k$ in \tref{etak.tight} can be understood by the following Poissonian ansatz. Consider $x_i$, $i=1,\ldots,N$, to be an i.i.d.\ family with $x_1$ having law with density proportional to $e^{-\beta \g(x)} dx$ on $B_1(0) \subset \R^d$. If we let $\eta_k$ be the $k$th smallest element of $\{|x_i| : i =1,\ldots,N\}$, then the order of $\eta_k$ agrees with \tref{etak.tight}.
	\end{remark}
	
	The proof of the upper bounds on $\eta_k$ in \tref{etak.tight} is an application of the ideas from the model computation in \sref{intro.model}, except with some more precision. The proof of the lower bounds can be understood via an idea related to isotropic averaging that we term ``mimicry." See the beginning of \sref{etak} for a high-level discussion of the proof.

	\subsection{Discrepancy and incompressibility bounds.} \label{s.intro_discrepancy}\
	Our previously mentioned results all concerned either sub-microscopic length scales or high particle densities, but we can in fact effectively use isotropic averaging on mesoscopic length scales and under only slight particle density excesses. We are interested in the discrepancy
	\begin{equation}
		\Disc(\Omega) = X(\Omega) - N\mueq(N^{-\frac{1}{d}} \Omega), \quad \Omega \subset \R^d
	\end{equation}
	which measures the deficit or excess of particles with reference to the equilibrium measure in a measurable set $\Omega$. It is also useful to define the {\it compression}
	\begin{equation}
		\Disc_{W}(\Omega) = X(\Omega) - \frac{1}{c_d} \int_{\Omega} \Delta W(x) dx
	\end{equation}
	where $c_d$ is such that $\Delta \g = -c_d \delta_0$. Note that $\Disc$ and $\Disc_W$ agree when $\Omega \subset \Sigma_N := N^{\frac1d}\supp \ \mueq$ and the equilibrium measure exists.
	
	In \cite{AS21}, it is proved that the discrepancy in $\Omega = B_R(z)$ is typically not more than $O(R^{d-1})$ in size whenever $z$ is sufficiently far in the interior of the droplet and $R$ is sufficiently large. The proof involves a multiscale argument inspired by stochastic homogenization \cite{AJM19}, as well as a technical screening procedure. More generally, \cite{AS21} gives local bounds on the electric energy, which provide a technical basis for central limit theorems \cite{S22} among other things. The upper bound of $O(R^{d-1})$ represents that the dominant error contribution within the argument comes from surface terms appearing in the multiscale argument. The JLM laws \cite{JLM93} predict that the discrepancy in $B_R(z)$ is actually typically of size $O(R^{(d-1)/2})$; in particular, the Coulomb gas is {\it hyperuniform}. This motivates the search for other methods to prove discrepancy bounds that, with enough refinement, may be able to overcome surface error terms.
	
	A motivation for studying the compression $\Disc_W$ comes from fractional quantum Hall effect (FQHE) physics (see \cite{R22,R22_2} and references therein for a more comprehensive discussion). In FQHE physics, one considers wave functions $\Psi_F : \mbb C^N \to \mbb C$ of the form (considering $x_i \in \mbb C \cong \R^2$)
	\begin{equation}
		\Psi_F(X_N) = F(X_N) \Psi_{\mathrm{Lau}}(X_N), \quad \Psi_{\mathrm{Lau}}(X_N) \propto \prod_{1 \leq i < j \leq N} (x_i - x_j)^\ell e^{-B\sum_{i=1}^N |x_i|^2/4 }
	\end{equation}
	for integers $\ell \geq 1$, symmetric, analytic functions $F : \mbb C^N \to \mbb C$, and magnetic field strengths $B > 0$. All wave functions are assumed to be $L^2(\mbb C^N)$ normalized. This means that $|\Psi_F(X_N)|^2 dX_N = \P^{W,U}_{N,\beta}(dX_N)$ for special values $W = W_{B,\ell}$, $\beta =  \ell$, and $$U(X_N) = -\beta^{-1} \log |F(X_N)|^2.$$ In particular, the function $U$ is superharmonic in each variable since $F$ is analytic. The connection between the {\it Laughlin wave function} $\Psi_{\mathrm{Lau}}$ and the Coulomb gas is termed the {\it plasma analogy}.
	
	When studying the robustness of FQHE under e.g.\ material impurities, a physically relevant variational problem is to find $F$ minimizing the functional
	\begin{equation}
	E_{\mcl O}(N,B,\ell) = \inf_{F} \E^{W,U}_{N,\beta} \left[ \sum_{i=1}^N \mcl O(x_i) \right] = \inf_F \int_{\R^2} \mcl O(x) \rho_{1,F}(x)
	\end{equation}
	for a certain $\mcl O : \R^2 \to \R$ giving a one-body energy associated to material impurities and trapping. Here $\P^{W,U}_{N,\beta}(dX_N) = |\Psi_F(X_N)|^2 dX_N$ as in the plasma analogy and $\rho_{1,F}$ is the $1$-point function of $\P^{W,U}_{N,\beta}$. It is expected that the infimum within $E_{\mcl O}(N,B,\ell)$ is approximately achieved as $N \to \infty$ with $F$ of the form $F(X_N) = \prod_{i=1}^N f(x_i)$ for an analytic function $f$. Such a factorization is important physically as it indicates the presence of uncorrelated ``quasi-holes" at the zeros of $f$, a remarkably simple system when compared to the those with nontrivial couplings between particles and quasi-holes.
	
	Lieb, Rougerie, and Yngvason, in the main result of \cite{LRY19}, proved $\rho_{1,F} \leq c_d^{-1} \Delta W (1+o(1))$ as $N \to \infty$ when integrated on scales of size $N^{1/4 + \ep}$ for $\ep > 0$. In other words, we have
	\begin{equation*}
		\Disc_W(B_R(z)) \leq o(1) \quad \text{as } N \to \infty
	\end{equation*}
	for $R \geq N^{1/4 + \ep}$ with high probability. Such a result is called an {\it incompressibility estimate}. A consequence is that if $\mcl O$ varies on a scale larger than $N^{1/4}$, then we have that the ``bathtub" energy
	\begin{equation} \label{e.bt.energy}
		E_{\mcl O}^{\mathrm{bt}}(N,B,\ell) = \inf \left \{ \int_{\R^2} \mcl O(x) \rho(x) \ : \ \int_{\R^2} \rho(x) dx = N, \ 0 \leq \rho \leq c_d^{-1} \Delta W \right \}
	\end{equation}
	is an approximate lower bound for $E_{\mcl O}(N,B,\ell)$. The infimum in \eref{bt.energy} is over $\rho$ that do not necessarily come as a $1$-point function for some $\P^{W,U}_{N,\beta}$. The restriction on the length scale of $\mcl O$ does not capture all physically realistic scenarios. To prove that $F(X_N) = \prod_{i=1}^N f(x_i)$ approximately saturates the infimum in $E_{\mcl O}(N,B,\ell)$, the remaining task is to show that a set of profiles $\rho$ saturating the infimum in the bathtub energy are well approximated by a $1$-point function $\rho_{1,F}$ with $F$ of the factorized form, a task that was considered in \cite{RY18, OR20}.
	
	We present a new method to prove incompressibility down to large microscopic scales, i.e.\ for $R \gg 1$, and also to give quantitative estimates for $o(1)$ terms. Our result is also interesting because it gives density upper bounds on balls $B_R(z)$ with weaker restrictions on the location of $z$ than in \cite{AS21}. For technical reasons appearing in \pref{energymin}, we will need to approximate the density of $\mu$ by a constant in $B_{2R}(z)$, and the resulting error begins to dominate for $R \geq N^{3/(5d)}$. We have therefore chosen to restrict to small enough mesoscales $R$. For FQHE applications, one has that $\mu$ is a multiple of Lebesgue measure and so this restriction is unnecessary.
	
	\begin{theorem}[Incompressibility] \label{t.fLL.over}
		Let $R \geq 1$ and $z \in \R^d$. Suppose either $\Delta W$ is constant in $B_{2R}(z)$ or both of the following: firstly that
		$$
		\| \nab \Delta W \|_{L^\infty(B_{2R}(z))} \leq CN^{-1/d},
		$$
		and secondly $R \leq N^{3/(5d)}$. Assume also $\Delta W(x) \geq C^{-1}$ for $x \in B_{2R}(z)$. Then we have
		\begin{equation} \label{e.fLL.over}
			\P^{W,U}_{N,\beta}(\{\Disc_W(B_R(z)) \geq T R^{d - 2/3}(1+\1_{d=2}\log R) \}) \leq e^{-c R^d T} + e^{-c R^{(d+2)/3} T^2} + e^{- R^{d+2}}
		\end{equation}
		for some $c > 0$ for all $T \geq 1$ large enough.
	\end{theorem}
	
	\tref{fLL.over} affirmatively answers an important question posed in \cite{LRY19}, demonstrating the remarkable property that the Coulomb gas cannot be significantly compressed beyond density $c_d^{-1}\Delta W$ by {\it any} choice of superharmonic perturbation $U$. The importance of incompressibility of the Laughlin phase was first raised in \cite{RY15_1} and first progress was made in \cite{RY15_2}. Further progress and the previously best result appears in \cite{LRY19}. There, the authors transfer $\beta = \infty$ incompressibility estimates to the positive temperature system, whereas we work directly with the positive temperature Gibbs measure. We also avoid completely the use of potential theoretic subharmonic quadrature domains (also termed ``screening regions"). We expect our result will have significant applications toward proving the stability of the Laughlin phase, which appears in the study of 2D electron gases and rotating Bose gases \cite{RSY14}.

	\begin{remark}
		\tref{fLL.over} may, at first, seem to be in disagreement with numerical results showing oscillations in $\rho_1$ at the edge of the droplet \cite{CSA20, C06, CW03} in $d=2$, with $\sup \rho_1$ significantly larger than $(2\pi)^{-1} \Delta W$. The oscillation wavelength appears to be of order of the inter-particle distance $R=1$, whereas our theorem becomes effective for $R \gg 1$, resolving the apparent disagreement.
	\end{remark}
	
	
	While \tref{fLL.over} controls only positive discrepancies, when combined with an estimate on the fluctuations of smooth linear statistics, it can also give lower bounds. Indeed, when $B_R(z)$ has a large discrepancy, the physically realistic scenario is that positive charge excess builds up either on the inside or outside of $\pa B_R(z)$. We call a thin annulus with the positive charge buildup a ``screening region", the existence of which is implied by rigidity for smooth linear statistics. We use rigidity from \cite{S22}, which we note does not take ``heavy-lifting", e.g.\ it is independent of the multi-scale argument from \cite{AS21}. \pref{fLL.upbd} proves a stronger version of \tref{fLL.over} that applies to screening regions.
	
	We will need some additional assumptions to apply results from \cite{S22} and \cite{AS19}. We refer to the introduction of \cite{S22} for commentary on the conditions. While the conditions hold in significant generality, our main purpose is to demonstrate the usefulness of \tref{fLL.over} rather than optimize the conditions. Assume that $W = V_N$ where $V \in C^{7}(\R^d)$, the droplet $\Sigma = \supp \ \mueq$ has $C^{1}$ boundary, $\Delta V \geq C^{-1} > 0$ in a neighborhood of $\Sigma$, and $\sup \Delta V \leq C$. Assume further
	\begin{equation*}
		\begin{cases}
			&\int_{\R^d} e^{-\frac{\beta}{2} N(V(x) - \log(1+|x|))} dx + \int_{\R^d} e^{-\beta N (V(x) - \log(1+|x|))}(|x| \log(1+|x|))^2 dx < \infty \quad \text{if }d=2,\\
			&\int_{\R^d} e^{-\frac{\beta}{2}V(x)} dx < \infty \quad \text{if }d=3, \\
			&\lim_{|x| \to \infty} V(x) + \g(x) = +\infty.
		\end{cases}
	\end{equation*}
	Finally, assume there exists a constant $K$ such that
	$$
	\g \ast \mueq(x) + V(x) - K \geq C^{-1} \min(\dist(x,\Sigma)^2,1) \quad \forall x \in \R^d.
	$$
	
	\begin{theorem} \label{t.fLL.improved}
		Let $R \in [1,N^{\frac{5}{7d}}]$ and $z \in \R^d$ be such that $B_{2R}(z) \subset \{x \in \Sigma_N \ : \ \dist(x,\pa\Sigma_N) \geq C_0 N^{1/(d+2)}\}$ for a large enough constant $C_0$. Assume that $W = V_N$ with $V$ satisfying the above conditions.
		Then we have
		$$
		\P^{V_N}_{N,\beta}(\{|\Disc(B_R(z))| \geq T R^{d-4/5}(1+\1_{d=2}\log R)\}) \leq e^{-cR^{d-10/15}T} + e^{-cR^{2/5 + d/15}T^2} + e^{-cR^{(d+2)/5}}.
		$$
		for large enough $T \gg 1$ and some $c > 0$.
	\end{theorem}
	We note that by applying isotropic averaging to a screening region, not only do we obtain bounds on the absolute discrepancy, we also give a sharper bound on the positive part than \tref{fLL.over}, albeit with some extra restrictions on $R$ and $z$.
	\subsection{Notation.}\ \label{s.notation}
	We now introduce some notation and conventions used throughout the paper. First, we recall the point process $X$ and ``index" process $\mbb X$ introduced at the end of \sref{intro1}. All point processes will be simple.
	
	Implicit constants $C$ will change from line to line and may depend on $\sup \Delta W$ and $d$ without further comment. In all sections except \sref{LL}, we will also allow $C$ to depend on continuously $\beta$ and $\beta^{-1}$. A numbered constant like $C_0, C_1$ will be fixed, but may be needed to be taken large depending on various parameters. For positive quantities $a,b$, we write $a \gg b$ to mean that $a \geq C b$ for a large enough constant $C > 0$, and $a \ll b$ for $a \leq C^{-1} b$ for large enough $C > 0$.
	
	For brevity we will sometimes write $\P$ for $\P^{W,U}_{N,\beta}$ and $\E$ for $\E^{W,U}_{N,\beta}$. This will only be done in proofs or sections where the probability measure is fixed throughout. We will write $\g(s)$ to mean $-\log s$ in dimension $2$ or $|s|^{-d+2}$ in $d \geq 3$ when $s > 0$. For a measure with a Lebesgue density, we often denote the density with the same symbol as the measure, e.g.\ $\nu(x)$ as the density of $\nu(dx)$.
	
	When it exists, we let $\mueq$ be the equilibrium measure associated to $V$, and let $\Sigma = \supp \ \mueq$ be the droplet. Note that $\mueq$ is a probability measure and $\Sigma$ has length scale $1$. We let $\mueq^N$ be the blown-up equilibrium measure with $\mueq^N(A) = N \mueq(N^{-1/d}A)$ for Borel sets $A$ and $\Sigma_N = \supp \ \mueq^N = N^{1/d} \supp \mueq$ be the blown-up droplet. Finally, we define
	\begin{equation}
		\mu(dx) = \frac{1}{c_d} \Delta W(x) dx
	\end{equation}
	where $c_d$ is such that $\Delta \g = -c_d \delta_0$. When we take $W = V_N$, we have that $\mu = \mueq^N$ on the blown-up droplet $\Sigma_N$.
	
	We let $B_R(z) \subset \R^d$ be the open Euclidean ball of radius $R$ centered at $z$, and for a nonempty interval $I \subset \R^{\geq 0}$, we let $\Ann_I(z)$ be the annulus containing points $x$ with $|x-z| \in I$. We use $|\cdot|$ to denote both Lebesgue measure for subsets of $\R^d$ and cardinality for finite sets. It will be clear via context which is meant.
	
	\begin{remark} \label{r.A1}
		For most applications, the requirement $\Delta W \leq C$ can be loosened to boundedness on the macroscopic scale. This is because the points $X_N$ are typically confined to a vicinity of the droplet and our arguments are localized to macroscopic neighborhoods of the regions to which they are applied. In $d=2$, for example, Ameur \cite{A21} has proved strong confinement estimates. Furthermore, one can apply our arguments to other Coulomb-type systems like finite volume jelliums provided our arguments do not ``run into" domain boundaries. The limiting factor is generally in iterating \pref{1C.LL.iterate} as in the proof of \tref{1C.LL}.
	\end{remark}
	
	\subsubsection*{Acknowledgments.}\ The author would like to thank his advisor, Sylvia Serfaty, for her encouragement and comments. He also thanks Nicolas Rougerie for helpful comments on the incompressibility estimates, especially in connection to the fractional quantum Hall effect. The author was partially supported by NSF grant DMS-2000205.
	
	\section{High Density Estimates} \label{s.LL}
	In this section, we prove \tref{1C.LL}. We will write $\P$ for $\P_{N,\beta}^{W,U}$, and implicit constants $C$ will depend only on $d$ and $\sup \Delta W$; they are independent of $\beta$.
	
	The idea behind the proof is as follows. We will consider two scales $r,R$ with $1 \leq r  = \lambda^{-1}R$ for $\lambda \geq 10$  and two concentric balls $B_r(z)$ and $B_R(z)$. In the event that $X(B_r(z)) \gg r^d$, there is a large pairwise Coulomb energy benefit upon replacing each point charge within $B_r(z)$ by annuli of scale $R$. In particular, this energy benefit dominates any  loss from the potential term $\sum_i W(x_i)$ in $\mcl H^{W,U}$. We must however consider entropy factors: after applying isotropic averaging, the particles originally confined to $B_r(z)$ become indistinguishable from the particles within $B_R(z)$. If $X(B_R(z)) \leq \lambda^d X(B_r(z))$, i.e.\ the density of particles in $B_R(z)$ is not larger than that of $B_r(z)$, the entropy costs are manageable. If, on the other hand, $B_R(z)$ has an extremely high particle density, we may iterate our estimate to larger scales $R$ and $\lambda R$. The iteration terminates once $R^d \gg N$ and the considered overcrowding event becomes impossible.
	
	\pref{1C.LL.iso} will compute the energy change upon isotropic averaging and estimate the adjoint isotropic averaging operator on the relevant event, which \pref{1C.LL.iterate} uses as in the model computation to obtain the iteration step. \tref{1C.LL} will then follow shortly.  Note that the below computation leading up to \eref{iso.W} and \eref{iso.subharmonic} will be used often in slightly different contexts.
	\begin{proposition} \label{p.1C.LL.iso}
		Consider $0 < r < \frac{1}{10} R$ and $\nu_R$ the uniform probability measure on the annulus $\Ann_{[\frac12 R, R-2r]}(0)$. We have
		\begin{equation} \label{e.LL.iso.energy}
			\Iso_{\mbb X(B_r(z)), \nu_R} \mcl H^{W,U}(X_N) \leq \mcl H^{W,U}(X_N) + CR^2X(B_r(z)) + \binom{X(B_r(z))}{2} \(\g\(\frac R2\)-\g(2r)\),
		\end{equation}
		and for any index sets $\N, \M \subset \{1,\ldots,N\}$ we have
		\begin{equation} \label{e.LL.adjoint}
			\Iso_{\N,\nu_R}^\ast \1_{\{\mbb X(B_r(z)) = \N\} \cap \{\mbb X(B_R(z)) = \M \}} \leq e^{C|\N|}\(\frac{r}{R}\)^{d|\N|} \1_{\{ \mbb X(B_R(z)) = \M \}}.
		\end{equation}
	\end{proposition}
	\begin{proof}
		We first prove \eref{LL.iso.energy} by considering the effect of the isotropic averaging operator on $\g(x_i - x_j)$ and $W(x_i)$. Let $\sigma$ denote $(d-1)$-dimensional Hausdorff measure. First, since
		$
		-\Delta \g = c_d \delta_0
		$
		for $c_d = \sigma(\pa B_1)(d-2 + \1_{d=2})$, we have for any $s > 0$ and $y \in \R^d$ by Green's third identity that
		\begin{align*}
			\lefteqn{ \frac{1}{c_d} \int_{B_s(y)} \g(x-y) \Delta W(x) dx + W(y) } \quad & \\ &= \frac{1}{c_d s} \int_{\pa B_s(y)}\(  \g(x-y) \nab W(x) \cdot (x-y) - W(x) \nab \g(x-y) \cdot (x-y) \) \sigma(dx).
		\end{align*}
		Using the divergence theorem, the RHS can be simplified further to
		$$
		\frac{\g(s)}{c_d} \int_{B_s(y)} \Delta W(x) dx + \frac{d-2 + \1_{d=2}}{c_d s^{d-1}} \int_{\pa B_s(y)}W(x) \sigma(dx).
		$$
		After a rearrangement, we see
		$$
		\frac{1}{\sigma(\pa B_s)} \int_{\pa B_s(y)}W(x) \sigma(dx) = W(y) + \frac{1}{c_d} \int_{B_s(y)} (\g(x-y) - \g(s)) \Delta W(x) dx,
		$$
		and we can then integrate against $\sigma(\pa B_s) ds$ with $y = x_1$ to see
		\begin{align*}
		\lefteqn{ \Iso_{\{1\},\nu_R} W(x_1) } \quad & \\ &= W(x_1) + \frac{1}{c_d |\Ann_{[\frac12 R, R-2r]}(0)|} \int_{\frac12 R}^{R-2r} \sigma(\pa B_s) \int_{B_s(x_1)} (\g(x-x_1) - \g(s)) \Delta W(x) dx ds.
		\end{align*}
		If we bound $\Delta W \leq C$, one can check by explicit integration that
		\begin{equation} \label{e.iso.W}
			\Iso_{\{1\},\nu_R} W(x_1) \leq W(x_1) + CR^2.
		\end{equation}
		A similar computation, this time using $\Delta \g \leq 0$ or superharmonicity of $U$ in each variable, shows that
		\begin{equation} \label{e.iso.subharmonic}
			\Iso_{\mbb X(B_r(z)), \nu_R} \g(x_i - x_j) \leq \g(x_i - x_j), \quad \Iso_{\mbb X(B_r(z)), \nu_R} U(X_N) \leq U(X_N)\quad \forall i,j \in \{1,\ldots,N\}.
		\end{equation}
		Finally, by Newton's theorem, the Coulomb interaction between a sphere of unit charge and radius $s$ and a point charge is bounded above by $\g(s)$. It follows from superposition that
		$$
		\Iso_{\mbb X(B_r(z)), \nu_R} \g(x_i - x_j) \leq \g\(\frac{R}{2}\)
		$$
		whenever $i \in \mbb X(B_r(z))$, particularly whenever $i,j \in \mbb X(B_r(z))$ in which case $\g(x_i - x_j) \geq \g(2r)$. Putting the above results together proves \eref{LL.iso.energy}.
		
		We now consider \eref{LL.adjoint}. Using that $\Iso_{\N,\nu_R} = \Iso_{\N,\nu_R}^\ast$ is a convolution on $(\R^d)^{\N}$ and Young's inequality, we have
		\begin{align} \label{e.isoptwise}
			\lefteqn{ \Iso_{\N,\nu_R}^\ast \1_{\{\mbb X(B_r(z)) = \N\} \cap \{\mbb X(B_R(z)) = \M \}} } \quad & \\ \notag &\leq \| \nu_R \|_{L^\infty(\R^d)}^{|\N|}  \|\1_{\{\mbb X(B_r(z)) = \N\} \cap \{\mbb X(B_R(z)) = \M \}} \|_{L^1((\R^d)^{\N})} \leq e^{C|\N|}\(\frac{r}{R}\)^{d|\N|}.
		\end{align}
		For any configuration $X_N$ with $\mbb X(B_R(z)) \ne \M$, we claim that $$\Iso_{\N,\nu_R}^\ast \1_{\{\mbb X(B_r(z)) = \N\} \cap \{\mbb X(B_R(z)) = \M \}}(X_N) = 0.$$ First, our claim follows if $\mbb X(B_R(z)) \setminus \N \ne \M\setminus \N$ simply because our isotropic averaging operator leaves coordinates with labels in $\mcl N^c$ fixed. So we may instead assume there exists $i \in \N$ with $x_i \not \in B_R(z)$. Then convolution with $\nu_R$ in the $x_i$ coordinate considers translates $x_i + y$ with $|y| \leq R-2r$, none of which can be found in $B_r(z)$. Our claim follows, and together with our pointwise bound \eref{isoptwise} this establishes \eref{LL.adjoint}.
	\end{proof}
	
	We are ready to prove the main iterative estimate that establishes \tref{1C.LL}.
	\begin{proposition} \label{p.1C.LL.iterate}
		Let $0 < r < R$ be such that $\lambda := \frac{R}{r} \geq 10$. Then we have that
		$$
		\P(\{X(B_r(z)) \geq Q \}) \leq 	\P(\{X(B_R(z)) \geq \lambda^d Q \}) + e^{C(1+\beta \lambda^2 r^2) Q - \beta \binom{Q}{2}(\g(2r) - \g(\lambda r/2))}
		$$
		for all $z \in \R^d$ and integers $Q$ with
		\begin{equation}
			Q \geq \begin{cases}
				\frac{C \lambda^2 r^2 + C\beta^{-1} }{\log(\frac14 \lambda)} \quad &\text{if } d=2,\\
				C \lambda^2 r^{d} +C \beta^{-1} r^{d-2} \quad &\text{if } d \geq 3.
			\end{cases}
		\end{equation}
	\end{proposition}
	\begin{proof}
		For simplicity, we consider $\lambda$ an integer. Let $\N \subset \M$ be index sets of size $n$ and $m$, respectively. We apply the model computation detailed in \sref{intro.model} and \pref{1C.LL.iso} to see
		\begin{align*}
		\lefteqn{ \P(\{\mbb X(B_r(z)) = \N\} \cap \{\mbb X(B_R(z)) = \M \}) } \quad & \\ &\leq e^{C(1+\beta R^2)n - \beta \binom{n}{2}(\g(2r) - \g(R/2))} \lambda^{-d n} \P(\{\mbb X(B_R(z)) = \M\}). 
		\end{align*}
		By particle exchangeability, we have
		$$
		\P(\{X(B_r(z)) = n\} \cap \{X(B_R(z) = m)\}) = \binom{N}{m} \binom{m}{n}\P(\{\mbb X(B_r(z)) = \N\} \cap \{\mbb X(B_R(z)) = \M \}),
		$$
		$$
		\P(\{X(B_R(z) = m)\}) = \binom{N}{m}\P(\{ \mbb X(B_R(z)) = \M\}),
		$$
		whence
		\begin{align*}
		\lefteqn{ \P(\{X(B_r(z)) = n\} \cap \{X(B_R(z) = m)\}) }\quad & \\ &\leq e^{C(1+\beta R^2)n - \beta \binom{n}{2}(\g(2r) - \g(R/2))} \binom{m}{n} \lambda^{-dn}\P(\{ X(B_R(z)) = m\}).
		\end{align*}
		By Stirling's approximation, we can estimate for $1 \leq n \leq m$ that
		$$
		\binom{m}{n} \leq 2\sqrt{\frac{m}{n(m-n)}}e^{n \log \frac{m}{n} + (m-n) \log \frac{m}{m-n}}\leq e^{Cn} e^{n \log \frac{m}{n}}.
		$$
		The RHS is bounded by $e^{Cn} \lambda^{dn}$ in the case that $m \leq \lambda^d n$, which we consider.
		We thus have
		\begin{align*}
			\lefteqn{ \P(\{X(B_r(z)) \geq Q \} \cap \{X(B_R(z)) \leq Q\lambda^d \})} \quad & \\ &\leq \sum_{n = Q}^{Q \lambda^d} e^{C(1+\beta R^2)n - \beta \binom{n}{2}(\g(2r) - \g(R/2))} \sum_{m = n}^{Q \lambda^d} \P(\{ X(B_R(z)) = m\}) \\
			&\leq \sum_{n = Q}^{Q \lambda^d} e^{C(1+\beta R^2)n - \beta \binom{n}{2}(\g(2r) - \g(R/2))}.
		\end{align*}
		We can bound the ratio between successive terms in the last sum above as
		$$
		e^{C(1+\beta R^2) - \beta n(\g(2r) - \g(R/2))} \leq \frac12
		$$
		if $Q \log(\lambda/4) \geq C\beta^{-1} + C r^2 \lambda^2$ in $d=2$ and if $Q \geq C \beta^{-1}r^{d-2} + C\lambda^2 r^d$ in $d \geq 3$. We conclude
		$$
		\P(\{X(B_r(z)) \geq Q \} \cap \{X(B_R(z)) \leq Q\lambda^d \}) \leq e^{C(1+\beta \lambda^2 r^2) Q - \beta \binom{Q}{2}(\g(2r) - \g(\lambda r/2))},
		$$
		and the proposition follows.
	\end{proof}
	
	We conclude the section with a proof of the high density JLM laws.
	\begin{proof}[Proof of \tref{1C.LL}]
		In $d=2$, we let $\lambda \geq 10$ be a free parameter, and in $d \geq 3$ we fix $\lambda$ large enough dependent on $d$. We apply \pref{1C.LL.iterate} iteratively to a series of radii $r_k = \lambda r_{k-1}$, $k \geq 1$, with $r_0 = R$ to achieve
		$$
		\P(\{X(B_R(z)) \geq Q \} ) \leq \limsup_{k \to \infty} \P(\{X(B_{r_k}(z)) \geq \lambda^{dk} Q \} ) + \sum_{k=0}^\infty e^{a_k}
		$$
		for
		$$
		a_k = {C(1+\beta \lambda^{2k+2} R^2) \lambda^{dk}Q - \beta { \binom{\lambda^{dk} Q}{2}}(\g(2\lambda^{k}R) - \g(\lambda^{k+1} R/2))}.
		$$
		If $d=2$, we compute
		$$
		a_{k+1} - a_k \leq -\beta \frac{\lambda^{4(k+1)} \log(\lambda/4) Q^2}{4} + C(1 + \beta \lambda^{2k+4}R^2)\lambda^{2(k+1)}Q.
		$$
		We can use $Q \log(\lambda/4) \geq C \lambda^2 R^2 + C \beta^{-1}$ for a large enough $C$ to see
		$$
		a_{k+1} - a_k \leq -\beta \frac{\lambda^{4(k+1)} \log(\lambda/4) Q^2}{8} \leq - \log 2.
		$$
		In $d \geq 3$, we have
		\begin{align*}
		a_{k+1} &\leq C(1+\beta \lambda^{2k+2} R^2) \lambda^{dk}Q-\beta \binom{\lambda^{dk}Q}{2} \frac{1}{2^{d-2}R^{d-2}\lambda^{(k+1)(d-2)}} \(1-\frac{2^{2d-4}}{\lambda^{d-2}}\) \\
		&\leq C(1+\beta \lambda^{2k+2} R^2) \lambda^{dk}Q-   \frac{\beta\lambda^{dk + 2k + 2- d} Q^2}{2^{d}R^{d-2}}.
		\end{align*}
		By using $Q \geq CR^d + C\beta^{-1}R^{d-2}$ (and $\lambda$ fixed), we find
		$$
		a_{k+1} \leq - \beta \frac{\lambda^{(d+2)k - d + 2} Q^2}{2^{d+1}R^{d-2}}.
		$$
		One can also compute
		$
		a_k \geq -\beta R^{-d+2}\lambda^{2dk - (d+2)k}Q^2
		$, which is dominated by $a_{k+1}$ if $\lambda$ is large enough. It follows that $a_{k+1} - a_k \leq -\log 2$.
		
		We obtain
		$$
		\P(\{X(B_R(z)) \geq Q \} ) \leq 2e^{a_0} \leq e^{- \beta {  \binom{Q}{2}}(\g(2R) - \g(\lambda R/2)) + C(1+\beta \lambda^2 R^2) Q}.
		$$
		The desired result for balls centered at a fixed $z$ follows from some routine simplifications.
		
		Finally, it will be useful to have versions of the overcrowding estimates for $z = x_1$. For this, note that conditioning $\P^{W,U}_{N,\beta}$ on $x_1$ gives a new measure $\P^{W,U+\sum_{i=2}^N \g(x_1 - x_i)}_{N-1,\beta}$ on $(x_2,\ldots,x_N)$. We can then apply our results to this $(N-1)$-particle Coulomb gas with modified potential. Actually, we could even extract an extra beneficial term in \eref{LL.iso.energy} from strict superharmonicity of $\g(x_1-\cdot)$, but it is mostly inconsequential for the large $Q$ results. We omit the details.
	\end{proof}

	\section{Clustering Estimates}
	The goal of this section is to prove \tref{kpoint} and \tref{1C.cluster}. Our idea is similar to that of the previous section, except we will work with submicroscopic scales and transport particles distances of order $1$. We will precisely compute energy and volume gains associated to the transport and control entropy costs using \tref{1C.LL} with $R=1$. Here, it will be important that we work with measures $\P^{W,U}_{N,\beta}$, since changing $U$ will effectively allow us to condition the gas without deteriorating the estimates. In this section, we will allow implicit constants $C$ to depend continuously on $\beta$, $\beta^{-1}$, and $\sup \Delta W$.
	
	\subsection{$k$-point function bounds}\
	We will first bound the $k$-point function $\rho_k(y_1,\ldots,y_k)$, $y_1,\ldots,y_k \in \R^d$, which was defined in \eref{kpoint.def}. In particular, we will prove \tref{kpoint}.
	
	Note that
	\begin{equation} \label{e.kpoint.iter}
		\rho_k(y_1,\ldots,y_k) = \rho_{k-1}(y_1,\ldots,y_{k-1}) \rho_1(y_k | y_1,\ldots,y_{k-1})
	\end{equation}
	where $\rho_1(\cdot|y_1,\ldots,y_{k-1})$ is the $1$-point function of the gas $\P^{W,U_{y_1,\ldots,y_{k-1}}}_{N-k+1,\beta}$ with
	$$
	U_{y_1,\ldots,y_{k-1}}(x_k,x_{k+1},\ldots,x_N) = U(y_1,\ldots,y_{k-1},x_k,\ldots,x_N) + \sum_{i=1}^{k-1} \sum_{j=k}^N \g(y_i - x_j).
	$$
	We let $X_{N,k} = (x_k,\ldots,x_N)$ and the representation
	\begin{equation} \label{e.1point.rep}
		\rho_1(y_k | y_1,\ldots,y_{k-1}) = \lim_{r \to 0^+} \frac{N-k+1}{|B_1(0)| r^d} \P^{W,U_{y_1,\ldots,y_{k-1}}}_{N-k+1,\beta}(\{x_k \in B_r(y_k)\})
	\end{equation}
	and bound the probability in the RHS using isotropic averaging. \tref{kpoint} will follow easily from iterating our one-point function bound. Let $\nu$ be the uniform probability measure on the annulus $\Ann_{[\frac12+2r,1-2r]}(0)$. We will replace $x_k$ by a charge shaped like $\nu$, and the below lemma gives estimates on the energy change and on the isotropic averaging operator.
	
	\begin{lemma} \label{l.kpoint.iso}
		With the definitions above and $0 < r < \frac{1}{100}$, we have
		\begin{align} \label{e.kpoint.delta}
			&\inf_{\substack{X_{N,k} \in \R^{d(N-k+1)}\\ x_1 \in B_r(y_k)}} \mcl H^{W, U_{y_1,\ldots,y_{k-1}}}(X_{N,k}) - \Iso_{\{k\},\nu} \mcl H^{W, U_{y_1,\ldots,y_{k-1}}}(X_{N,k})  \\\notag &\quad \geq -Ck + \sum_{i=1}^{k-1} \max(\g(|y_k - y_i| + r),0).
		\end{align}
		We also have
		\begin{equation} \label{e.kpoint.adj}
			\Iso^\ast_{\{k\},\nu} \1_{B_r(y_k)}(x_k) \leq Cr^d \1_{B_1(y_k)}(x_k).
		\end{equation}
	\end{lemma}
	\begin{proof}
		We begin by noting that (see \pref{1C.LL.iso} for a similar computation)
		\begin{equation*}
			\Iso_{\{k\},\nu} \g(x_i - x_j) \leq \g(x_i - x_j), \quad 	\Iso_{\{k\},\nu} W(x_k) \leq W(x_k) + C.
		\end{equation*}
		It follows that
		\begin{equation*}
			\Iso_{\{k\},\nu} \mcl H^{W,0}(X_{N,k}) \leq \mcl H^{W,0}(X_{N,k}) + C,
		\end{equation*}
		and it remains to consider the $U_{y_1,\ldots,y_{k-1}}$ term. For this term, we compute
		$$
		\Iso_{\{k\},\nu} \g(y_i - x_k) \leq C, \quad \Iso_{\{k\},\nu} \g(y_i - x_k) = \g(y_i - x_k) \text{ if } |y_i - y_k| \geq 1
		$$
		using Newton's theorem. We also have $\g(y_i - x_k) \geq \g(|y_k - y_i| + r)$ since $x_k \in B_r(y_k)$. Thus we have
		\begin{align*}
			\Iso_{\{k\},\nu} U_{y_1,\ldots,y_{k-1}}(X_{N,k}) &\leq U_{y_1,\ldots,y_{k-1}}(X_{N,k}) + \sum_{i=1}^{k-1} \1_{|y_i - y_k| \leq 1}(C - \g(y_i - x_k)) \\ &\leq U_{y_1,\ldots,y_{k-1}}(X_{N,k}) + Ck - \sum_{i=1}^{k-1} \max(\g(|y_k - y_i| + r,0)).
		\end{align*}
		This finishes the proof of \eref{kpoint.delta}. The proof of \eref{kpoint.adj} is straightforward using that $\Iso^\ast_{\{k\},\nu}$ is convolution by $\nu$.
	\end{proof}
	\begin{proposition} \label{p.1point}
		We have
		\begin{equation}
			\rho_1(y_k | y_1,\ldots,y_{k-1}) \leq e^{Ck -\beta \sum_{i=1}^{k-1} \max(\g(y_i - y_k),0)}
		\end{equation}
		for the one-point function of $\P^{W,U_{y_1,\ldots,y_{k-1}}}_{N-k+1,\beta}$.
	\end{proposition}
	\begin{proof}
		An application of the model computation with the operator $\Iso_{\{k\},\nu}$ and \lref{kpoint.iso} shows
		\begin{equation} \label{e.1point.apply}
			\P^{W,U_{y_1,\ldots,y_{k-1}}}_{N-k+1,\beta}(\{x_k \in B_r(y_k)\}) \leq r^d e^{Ck - \beta \sum_{i=1}^{k-1} \max(\g(y_i - y_k + r),0) } \P^{W,U_{y_1,\ldots,y_{k-1}}}_{N-k+1,\beta}(\{x_k \in B_1(y_k)\})
		\end{equation}
		for all $r$ small enough. Let $Y = \sum_{i=k}^N \delta_{x_i}$. Since
		$$
		\sum_{i=k}^N \1_{\{x_i \in B_1(y_k)\}} \leq \sum_{n=1}^{N-k+1} n \1_{\{Y(B_1(y_k)) = n\}}
		$$
		and by \tref{1C.LL} we have
		$$
		P^{W,U_{y_1,\ldots,y_{k-1}}}_{N-k+1,\beta}(\{Y(B_1(y_k)) = n\}) \leq Ce^{-c n^2}
		$$
		for some $c > 0$, we find by exchangeability that
		$$
		\P^{W,U_{y_1,\ldots,y_{k-1}}}_{N-k+1,\beta}(\{x_k \in B_1(y_k)\}) \leq \frac{C}{N-k+1} \sum_{n=1}^\infty n e^{-cn^2} \leq \frac{C}{N-k+1}.
		$$
		Combining this with \eref{1point.apply} and taking $r \to 0^+$ in \eref{1point.rep} proves the proposition.
	\end{proof}
	
	\begin{proof}[Proof of \tref{kpoint}.]
		The theorem follows easily by applying \pref{1point} iteratively to the representation \eref{kpoint.iter}
	\end{proof}
	
	\subsection{Bounds on particle clusters.}\
	We apply our $k$-point function estimates to prove the clustering estimates of \tref{1C.cluster}.
	
	\begin{proof}[Proof of \tref{1C.cluster}.]
		We will need slightly different arguments based on $d=2$ or $d \geq 3$. We integrate the results of \tref{kpoint}. We have
		\begin{align*}
			\P^{W,U}_{N,\beta}(\{X(B_r(z)) \geq Q\}) &\leq \frac{1}{Q!} \sum_{i_1,\ldots,i_Q \ \text{distinct}} \P^{W,U}_{N,\beta}\(\bigcap_{\ell=1}^Q \{ x_{i_\ell} \in B_r(z) \} \) \\ &= \frac{1}{Q!} \int_{(B_r(z))^{Q}} \rho_Q(y_1,\ldots,y_Q) dy_1 \cdots dy_Q.
		\end{align*}
		Thus, for a $Q$ dependent constant $C$, we have
		$$
		\P^{W,U}_{N,\beta}(\{X(B_r(z)) \geq Q\}) \leq C r^{Qd} \| \rho_Q \|_{L^\infty((B_r(z))^Q)}.
		$$
		In $d=2$ one estimates
		$$
		\| \rho_Q \|_{L^\infty((B_r(z))^Q)} \leq C r^{\beta {\binom{Q}{2}}}
		$$
		which establishes \eref{cluster.fixed.2}. In $d \geq 3$, we can estimate
		$$
		\| \rho_Q \|_{L^\infty((B_r(z))^Q)} \leq C e^{-\beta {\binom{Q}{2}}\inf_{x,y \in B_r(z)} \g(x-y)} \leq C e^{-\frac{\beta}{2^{d-2}r^{d-2}} {\binom{Q}{2}}},
		$$
		establishing \eref{cluster.fixed.3}. To prove the estimates with $z = x_1$, we consider $(x_2,\ldots,x_N)$ drawn from the Coulomb gas $\P^{W,U_{x_1}}_{N-1,\beta}$ where $U_{x_1}(x_2,\ldots,x_N) = U(X_N) + \sum_{i=2}^N \g(x_i - x_1)$. The argument of \tref{kpoint} applies to this gas, except we have an extra term when applying our isotropic averaging operator to $U_{x_1}$ coming from particle repulsion generated by $x_1$. It is a straightforward modification to include this term in \lref{kpoint.iso} and the proof of \tref{kpoint}. We can prove
		\begin{equation}
			\rho_{Q-1}(y_2,\ldots,y_Q | y_1) \leq \begin{cases}
				C\prod_{1 \leq i < j \leq Q} \min(1, |y_i - y_j|^\beta) \quad &\text{if}\ d=2, \\
				C e^{-\beta \mcl H^0(y_1,\ldots,y_Q)} \quad &\text{if}\ d \geq 3,
			\end{cases}
		\end{equation}
		where $\rho_{Q-1}(y_2,\ldots,y_Q | y_1)$ is the $(Q-1)$-point function of $\P^{W,U_{x_1}}_{N-1,\beta}$ with $x_1 = y_1$. Then since
		$$
		\P^{W,U}_{N,\beta}(\{X(B_r(x_1)) \geq Q\}) \leq C r^{(Q-1)d} \sup_{y_1 \in \R^d} \| \rho_{Q-1}(\cdot | y_1) \|_{L^\infty((B_r(y_1))^{Q-1})},
		$$
		we obtain
		$$
		\P^{W,U}_{N,\beta}(\{X(B_r(x_1)) \geq Q\}) \leq Cr^{(Q-1)d + \beta {\binom{Q}{2}}} \quad \text{if}\ d= 2,
		$$
		and
		\begin{equation}\label{e.cluster.centered.naive.3}
			\P^{W,U}_{N,\beta}(\{X(B_r(x_1)) \geq Q\}) \leq C r^{(Q-1)d} e^{-\beta \frac{Q-1}{r^{d-2}} - \beta \frac{1}{2^{d-2} r^{d-2}} {\binom{Q-1}{2}}} \quad \text{if}\ d\geq 3.
		\end{equation}
		In the latter estimate, we used that $|y_i - y_1| \leq r$ in our $L^\infty(B_r(y_1)^{Q-1})$ bound on $\rho_{Q-1}(\cdot|y_1)$.
		
		We expect one can obtain significant improvements to the $d \geq 3$ bound by more accurately estimating the minimum value of $\mcl H^0(y_1,\ldots,y_Q)$ and the relative volume within $(B_r(z))^Q$ of near-minimizers. We will now do so for the case of $Q=2$ and $z = x_1$ since it is relevant to the minimal separation problem.
		
		Let $d \geq 3$. We find
		\begin{align}
			\P^{W,U}_{N,\beta}(\{X(B_r(x_1)) \geq 2\}) \leq  C \sup_{y_1 \in \R^d} \int_{B_r(y_1)} e^{-\frac{\beta}{|y_1 - y_2|^{d-2}}} dy_2 &\leq C r^d \int_{B_1(0)} e^{-\frac{\beta}{r^{d-2} |y|^{d-2}}} dy \\ \notag &\leq Cr^d \int_0^1 s^{d-1} e^{-\frac{\beta}{r^{d-2}} \frac{1}{s^{d-2}}} ds.
		\end{align}
		For any $\alpha > 0$, we can compute
		$$
		\int_0^1 s^{d-1} e^{-\frac{\alpha}{s^{d-2}}} ds \leq \int_0^1 \frac{1}{s^{d-1}} e^{-\frac{\alpha}{s^{d-2}}} ds = \frac{e^{-\alpha}}{\alpha}.
		$$
		Applying this at $\alpha = r^{-d+2} \beta$ allows us to conclude
		$$
		\P^{W,U}_{N,\beta}(\{X(B_r(x_1)) \geq 2\}) \leq Cr^{2d-2} e^{-\frac{\beta}{r^{d-2}}}.
		$$
		Notice that this is significantly stronger than \eref{cluster.centered.naive.3} with $Q=2$. The ``volume factor" has been reduced from $r^d$ to $r^{2d-2}$.
	\end{proof}
	
	\section{Minimal Separation} \label{s.etak}
	In this section, we prove the minimal separation theorem \tref{etak.tight} via \pref{etak.isbig} and \pref{etak.issmall}. The proof that $\eta_k$ is not smaller than expected is a relatively straightforward application of the clustering bounds of \tref{1C.cluster}. The proof the $\eta_k$ is not too large takes a new idea that we term ``mimicry", for which we now give some intuition.
	
	Consider a configuration $X_N$ with $\eta_1 > r$, let $R = O(1)$ be a large microscopic scale, and let $r < R$. Consider two particles, say $x_1$ and $x_2$, with $|x_1 - x_2| \in (r,R)$. We can move the particles closer together by ``replacing" $x_2$ with a charged annulus of inner radius $r/2$ and outer radius $r$ centered $x_1$, and applying our model computation to associate a family of new point configurations to $X_N$. A key difference from isotropic averaging is that the annulus is centered at $x_1$ instead of $x_2$.
	
	By Newton's theorem, the interaction of the annulus with the particles $x_3,\ldots,x_N$ is the same or more mild than that of $x_1$ and the potential term is less than $W(x_1) + Cr^2$. The interaction between $x_1$ and $x_2$ has increased, in the worst case, from $\g(R)$ to $\g(r/2)$. The particle $x_2$ is ``mimicking" $x_1$. We can also have $x_1$ mimic $x_2$; one of the two mimicking operations is favorable energy-wise up to $O(r^2) + \g(r/2) - \g(R)$. There is also an entropy cost associated to mimicry. The particle $x_2$ originally occupies a $O(R^d)$ volume region around $x_1$ and afterwards is restricted to $O(r^d)$ volume. We can apply this argument to any of the $O(N)$ particle pairs within distance $R$, creating new configurations from $X_N$ with either (1) all particles separated by $r$ except a single pair or (2) with a cluster of three or more particles within a ball of radius $r$. Situation (2) can be proved unlikely using our clustering estimates from \tref{1C.cluster}. Assuming situation (1), we can never create the same configuration by applying a mimicry argument to two distinct index pairs, and so we achieve a volume benefit factor of $O(N)$. We find that $\eta_1 > r$ is unlikely as soon as
	$$
	\frac{Nr^d}{R^d} e^{-\beta (\g(r/2) - \g(R))} \gg 1.
	$$
	In $d=2$, this happens as soon as $r \ll N^{1/(2+\beta)}$, matching our desired result \tref{etak.tight}. In $d \geq 3$, we must be more careful and consider thinner annuli for mimicry, but a similar intuition holds. This argument is carried out in \pref{etak.issmall}, but before then we must take some care to provide an appropriate parameter $R$, which is done in \lref{lonely} using results from \cite{CHM18}. The main idea is that most of the particles are contained within some volume of size $O(N)$ with high probability.
	
	\begin{proposition} \label{p.etak.isbig}
		There is an absolute constant $C_0 > 0$ such that in $d=2$ we have
		\begin{equation} \label{e.etak.isbig.2}
			\P^{W,U}_{N,\beta}(\{\eta_k \leq \gamma N^{-\frac{1}{2+\beta}} \}) \leq C \gamma^{(2+\beta)k} + CN^{-\frac{2+2\beta}{2+\beta}} \gamma^{4+3\beta} + CN^{-7k}\gamma^{C_0(1+\beta)k} \quad \forall \gamma > 0,
		\end{equation}
		and in $d=3$ we have
		\begin{equation} \label{e.etak.isbig.3}
			\P^{W,U}_{N,\beta}\left (\left \{\eta_k \leq \(\frac{\beta}{\log N - \frac{2d-2}{d-2} \log \log N +\gamma}\)^{\frac{1}{d-2}}\right\} \right) \leq \frac{C}{\sqrt{N}} + Ce^{-k\gamma} \quad \forall \gamma > 0.
		\end{equation}
		The constant $C$ depends on $\beta$, $\sup \Delta W$, and $k$.
	\end{proposition}
	\begin{proof}		
		Let $r \in (0,1)$. Let $S$ be the event that there exists $i \in \{1,\ldots,N\}$ such that $B_{r}(x_i)$ contains $3$ or more particles $x_j$. We have
		\begin{equation} \label{e.3cluster.bd}
			\P^{W,U}_{N,\beta}(S) \leq CN\P^{W,U}_{N,\beta}(\{X(B_r(x_1)) \geq 3\}) \leq \begin{cases} CNr^{4+3\beta} \quad &\text{if}\ d=2, \\
				C N r^{2d} e^{-\frac{2\beta}{r^{d-2}}} \quad &\text{if}\ d\geq 3, \end{cases}
		\end{equation}
		by a union bound and \tref{1C.cluster}.
		
		Let $E_{k,r}$ be the event that $\eta_k \leq r$, and let $a_j< b_j$ be random indices such that $\eta_j = |x_{a_j} - x_{b_j}|$ for $j=1,\ldots,k$, which are well-defined almost surely. On the event $S^c \cap E_{k,r}$, we have that $\{a_j, b_j\} \cap \{a_{\ell}, b_\ell\} = \varnothing$ for $\ell \ne j$ almost surely. Let $F_j = \bigcap_{\ell=1}^j \{a_\ell = 2\ell - 1 \} \cap \{b_{\ell} = 2\ell\}$. By exchangeability, we have
		$$
		\P^{W,U}_{N,\beta}(E_{k,r} \cap S^c) \leq {\binom{N}{2}}^{k} \P^{W,U}_{N,\beta}(E_{k,r} \cap F_k) \leq {\binom{N}{2}}^{k} \P^{W,U}_{N,\beta}\(\bigcap_{\ell=1}^k \{|x_{2\ell} - x_{2 \ell-1}| \leq  r\}\) .
		$$
		Let $Y = X - \sum_{i=1}^{2k-2} \delta_{x_i}$, i.e.\ the point process with the first $2k-2$ points deleted. We estimate
		\begin{align} \label{e.etak.cond}
			\lefteqn{ \P^{W,U}_{N,\beta}\(\bigcap_{\ell=1}^k \{|x_{2\ell} - x_{2 \ell-1}| \leq  r \}\) } \quad & \\ \notag &= \P^{W,U}_{N,\beta}\(\bigcap_{\ell=1}^{k-1} \{|x_{2\ell} - x_{2 \ell-1}| \leq  r \}\) \P^{W,U}_{N,\beta}\(\{|x_{2k} - x_{2k-1}| \leq r\}\ \bigg | \ \bigcap_{\ell=1}^{k-1} \{|x_{2\ell} - x_{2 \ell-1}| \leq  r \} \),
		\end{align}
		and furthermore
		\begin{equation} \label{e.etak.cond2}
		\P^{W,U}_{N,\beta}\(\{|x_{2k} - x_{2k-1}| \leq r\} \ \bigg | \ \bigcap_{\ell=1}^{k-1} \{|x_{2\ell} - x_{2 \ell-1}| \leq  r \}\)  = \E \left[ \P^{W,U_{k-1}}_{N-2k+2,\beta}(\{ |x_{2k} - x_{2 k-1}| \leq  r \} ) \right].
		\end{equation}
		for the random potential 
		$
		U_{k-1}(x_{2k-1},x_{2k},\ldots,x_N) = U(X_N) + \sum_{j=2k-1}^N \sum_{\ell = 1}^{2k-2} \g(x_{\ell} - x_j).
		$ The expectation $\E$ is over the law of the points $x_{\ell}$, $\ell \in \{1,\ldots,2k-2\}$, conditioned to be pairwise close as above, and $\P^{W,U_{k-1}}_{N-2k+2,\beta}$ is a measure on the particles $(x_{2k-1},x_{2k},\ldots,x_N)$.
		
		For any positive integer $n$, we have by exchangeability that
		\begin{align*}
		\lefteqn{ \P^{W,U_{k-1}}_{N-2k+2,\beta}(\{|x_{2k} - x_{2 k-1}| \leq  r \}) } \quad & \\ &\leq \frac{C_n}{N} \P^{W,U_{k-1}}_{N-2k+2,\beta}(\{ Y(B_r(x_{2k})) \geq 2 \}) + \P^{W,U_{k-1}}_{N-2k+2,\beta}(\{ Y(B_r(x_{2k})) \geq n \}).
		\end{align*}
		We apply \tref{1C.cluster} to bound each piece above. Collecting the estimates, \eref{etak.cond2}, and \eref{etak.cond}, we have proved
		\begin{align*}
			\lefteqn{ \P^{W,U}_{N,\beta}\(\bigcap_{\ell=1}^k \{|x_{2\ell} - x_{2 \ell-1}| \leq  r \}\) } \quad & \\ &\leq \begin{cases} C \(\frac{1}{N} r^{2+\beta} + r^{2(n-1) + \beta {\binom{n}{2}}}\)\P^{W,U}_{N,\beta}\(\bigcap_{\ell=1}^{k-1} \{|x_{2\ell} - x_{2 \ell-1}| \leq  r \}\) \quad \text{if } d=2, \\
				C \(\frac{1}{N} r^{2d-2} e^{-\frac{\beta}{r^{d-2}}} + r^{(n-1)d} e^{-\frac{c\beta n^2}{r^{d-2}}} \)\P^{W,U}_{N,\beta}\(\bigcap_{\ell=1}^{k-1} \{|x_{2\ell} - x_{2 \ell-1}| \leq  r \}\) \quad \text{if } d\geq 3.
			\end{cases}
		\end{align*}
		for a dimensional constant $c > 0$.
		We can iterate this estimate to see
		\begin{equation} \label{e.etak.isbig.iter}
			\P^{W,U}_{N,\beta}\(\bigcap_{\ell=1}^k \{|x_{2\ell} - x_{2 \ell-1}| \leq  r \}\)  \leq  \begin{cases} C\(\frac{1}{N} r^{2+\beta} + r^{2(n-1) + \beta {\binom{n}{2}}}\)^k \quad &\text{if}\ d=2, \\
				C \(\frac{1}{N} r^{2d-2} e^{-\frac{\beta}{r^{d-2}}} + r^{(n-1)d} e^{-\frac{c\beta n^2}{r^{d-2}}} \)^k \quad &\text{if}\ d\geq 3. \end{cases}
		\end{equation}
		We conclude the general argument by writing
		\begin{equation} \label{e.etak.isbig.conclude}
			\P^{W,U}_{N,\beta}(E_{k,r}) \leq \P^{W,U}_{N,\beta}(S) + \P^{W,U}_{N,\beta}(E_{k,r} \cap S^c) \leq \P^{W,U}_{N,\beta}(S) + \binom{N}{2}^k \P^{W,U}_{N,\beta}\(\bigcap_{\ell=1}^k \{|x_{2\ell} - x_{2 \ell-1}| \leq  r \}\).
		\end{equation}
		The probability of $S$ is bounded in \eref{3cluster.bd}.
		
		Next, we choose specific $r$ and $n$ to conclude the proposition. We must consider $d=2$ and $d\geq 3$ separately. For $d=2$, we choose $r = \gamma N^{-\frac{1}{2+\beta}}$ and $n = 10$ to see
		\begin{align*}
		&\P^{W,U}_{N,\beta}(S) \leq CN^{-\frac{2+2\beta}{2+\beta}} \gamma^{4+3\beta}, \\ &\P^{W,U}_{N,\beta}\(\bigcap_{\ell=1}^k \{|x_{2\ell} - x_{2 \ell-1}| \leq  r \}\) \leq CN^{-2k} \gamma^{(2+\beta)k} + C N^{-9k} \gamma^{C_0(1+\beta)k}
		\end{align*}
		for some $C_0 > 0$, and \eref{etak.isbig.2} follows.
		
		For $d \geq 3$, we will assume $r^{d-2} \leq  \frac{4\beta}{3\log N}$, which with our choice below will happen if $N \geq C$, so that the first summand in the RHS of \eref{etak.isbig.iter} dominates for $n$ large enough. It follows that
		$$
		\P^{W,U}_{N,\beta}\(\bigcap_{\ell=1}^k \{|x_{2\ell} - x_{2 \ell-1}| \leq  r \}\) \leq C\frac{r^{(2d-2)k}}{N^k} e^{-\frac{\beta k}{r^{d-2}}},
		$$
		if $n$ is chosen large enough, and
		$$
		\P^{W,U}_{N,\beta}(S) \leq CN \exp\(-2\beta\(\frac{3\log N}{4\beta}\)\) \leq CN^{-\frac12}.
		$$
		We can then let $r = \(\frac{\beta}{\log N  - \frac{2d-2}{d-2}\log \log N + \gamma}\)^{\frac{1}{d-2}}$, and after a short computation using \eref{etak.isbig.conclude}, we find \eref{etak.isbig.3} holds.
	\end{proof}
	
	Our next goal is to prove that $\eta_k$, properly normalized, is tight as $N \to \infty$. To do this, we must create close particle pairs using the model computation, except with a new class of operators, and precisely estimate energy and entropy costs.
	
	For $i \in \{1,\ldots,N\}$, let $L_{i,R}$ be the event that $\min_{j \ne i} |x_j - x_i| \geq R$, i.e.\ that the point $x_i$ is ``lonely". Clearly, if $R$ is much larger than the interstitial distance, the event $L_{i,R}$ is rare. We use \cite{CHM18} to quantify this fact. The below result is their Theorem 1.5, written in our blown-up coordinates and corresponding to their inverse temperature chosen as $N^{-1+2/d} \beta$ in terms of our $\beta$.
	\begin{theorem}[\cite{CHM18}, Theorem 1.5]
		Let $W = V_N$ with $V$ satisfying \eref{A1} in $d=2$ and \eref{A2} in $d \geq 3$. Recall the blown-up equilibrium measure $\mueq^N$ from \sref{notation}. We have
		\begin{equation} \label{t.CHM}
			\P^{V_N}_{N,\beta}(\{d_{\mathrm{BL},N}(X,\mueq^N) \geq N\sqrt{\log N} r\}) \leq e^{-c N \log N r^2}
		\end{equation}
		for any $r \geq c^{-1}$ for some $\beta$-dependent $c > 0$, where, for nonnegative measures $\nu_1$ and $\nu_2$ of mass $N$, we define
		\begin{equation*}
			d_{\mathrm{BL},N}(\nu_1,\nu_2) := \sup_{\substack{ f \in C_N^{0,1}(\R^d) \\ \| f \|_{C_N^{0,1}(\R^d)} \leq 1}} \int_{\R^d} f(x)(\nu_1 - \nu_2)(dx),
		\end{equation*}
		and $C_N^{0,1}(\R^d)$ is the space of bounded Lipschitz functions with norm
		$$
		\| f \|_{C^{0,1}(\R^d)} = \sup_{x \in \R^d} N^{-1/d}|f(x)| + \sup_{x \ne y \in \R^d} \frac{|f(x) - f(y)|}{|x-y|}.
		$$
	\end{theorem}
	
	\begin{lemma} \label{l.lonely}
		Let $W = V_N$ with $V$ satisfying \eref{A1} in $d=2$ and \eref{A2} in $d \geq 3$. We have
		$$
		\P^{V_N}_{N,\beta}(L_{1,R}) \leq C R^{-d} + CN^{-\frac1d} \sqrt{\log N}
		$$
		for all $R \geq 1$.
	\end{lemma}
	\begin{proof}
		We will mostly omit $V_N, N,\beta$ from the notation. 
		
		Defining $\varphi(x) = \max(0,1-N^{-1/d}\dist(x,\supp \ \mueq^N))$, one can easily check
		$$
		\| \varphi(x) \|_{C^{0,1}_N(\R^d)} \leq CN^{-1/d}.
		$$
		Therefore
		$$
		\left | \int_{\R^d} \varphi(x) (X - \mueq^N)(dx) \right | \leq \min\(CN^{-\frac1d} d_{\mathrm{BL},N}(X,\mueq^N), N\),
		$$
		and so for $r = c^{-1}$ large enough we have by \tref{CHM}
		\begin{align} \label{e.CHM}
			\E\left[ \left | \int_{\R^d} \varphi(x) (X - \mueq^N)(dx) \right | \right] &\leq N \P(\{d_{\mathrm{BL},N}(X,\mueq^N) \geq N\sqrt{\log N} r\}) + CN^{1-\frac1d} \sqrt{\log N} r \\ \notag
			&\leq Ne^{-cN \log N r^2} +  CN^{1-\frac1d}\sqrt{\log N} r \leq CN^{1-\frac1d}\sqrt{\log N}.
		\end{align}
		
		We will use \eref{CHM} to show that most points are typically within $\supp \ \varphi$. Indeed, we have
		$$
		\E\left[\int_{\R^d} \varphi(x)(X - \mueq^N)(dx) \right] \leq 	\E\left[-X((\supp \ \varphi)^c) \right] = - \sum_{i=1}^N \P(\{x_i \not \in \supp \ \varphi\}),
		$$
		and so $\P(\{x_1 \not \in \supp \ \varphi\})  \leq CN^{-\frac1d}\sqrt{\log N}$.
		
		Let $\xi_i = \min(R,\min_{j \ne i} |x_i - x_j|)$, so the event $L_{i,R}$ is equivalent to $\xi_r = R$. Since the balls $B_{\xi_i/2}(x_i)$ are disjoint, we must have
		$$
		|\{i : x_i \in \supp \ \varphi, \xi_i = R\}| \cdot R^d \leq C|\{x : \dist(x,\supp \ \varphi) \leq R\}| \leq CN,
		$$
		where we assumed WLOG $R \leq N^{\frac1d}$ and used $\supp \ \mueq$ compact in the last inequality. Thus
		$$
		\P(L_{1,R} \cap \{x_1 \in \supp \ \varphi\}) \leq CR^{-d},
		$$
		and the lemma follows.
	\end{proof}
	
	For a configuration $X_N$, let $\phi(i) = \phi_{X_N}(i) \in \{1,\ldots,N\}$ be the index of the closest particle to $i$, i.e.\ $|x_i - x_{\phi(i)}| = \min_{i \ne j} |x_i - x_j|$. This is almost surely well-defined. Also, let $T_{k,r}$ be the event that the cardinality of $\{\{i,j\} \ : \ |x_i - x_j| < r, i \ne j\}$ is at most $k$. Define the random indices $a_\ell < b_\ell$ such that $\eta_\ell = |x_{a_\ell} - x_{b_{\ell}}|$. The next proposition uses a mimicry operator to move $x_i$ and $x_{\phi(i)}$ closer together.
	
	\begin{proposition} \label{p.createpairs}
		Suppose $U(X_N) = \sum_{i=1}^N U_1(x_i)$ for a superharmonic function $U_1 : \R^d \to \R$. Let $r \in (0,1)$ and let $\nu$ be a rotationally symmetric probability measure supported in $\ov{B_r(0)}$ with a Lebesgue density, also denoted $\nu$. Abbreviate $\P = \P^{W,U}_{N,\beta}$. For any $R \geq 1$, integer $n \geq 3$, and $k \in \{1,\ldots,N\}$, we have that $\P(\{\eta_k \geq r\})$ is bounded by
		\begin{align*}
			&{C e^{C \beta r^2 + \beta \Delta_{\nu}}} \| \nu \|_{L^\infty(\R^d)} M_{r,R}\(\frac{k+n}{N} + \P(\{X(B_r(x_1)) \geq 3\}) + N\P(\{X(B_r(x_1)) \geq n\})\) \\ &\quad + \frac{2k-2}{N} + \P(L_{1,R}),
		\end{align*}
		where $\Delta_\nu = \int_{\R^d} \g(y) \nu(dy)$ and $M_{r,R} = \int_{\Ann_{[r,R]}(0)} e^{-\beta \g(y)} dy$. The constant $C$ depends on $\sup \Delta W$ and $n$.
	\end{proposition}
	\begin{proof}
		We will abbreviate $\P = \P^{W,U}_{N,\beta}$ and $\E = \E^{W,U}_{N,\beta}$ throughout the proof.
		
		First note that $T_{k-1,r} = \{\eta_k \geq r\}$. Furthermore, on $T_{k-1,r}$ we have $|x_1 - x_{\phi(1)}| \geq r$ unless $a_\ell = 1$ or $b_\ell = 1$ for some $\ell \in \{1,\ldots,k-1\}$. We will furthermore want to fix the label $\phi(1)$ and ensure $x_1$ is not $R$-lonely, which inspires the bound
		\begin{align} \label{e.Tkr.breakdown}
			\P(T_{k-1,r})  &\leq (N-1)\P(\{|x_1 - x_2| \geq r\} \cap \{\phi(1) = 2\} \cap L_{1,R}^c \cap T_{k-1,r}) + \P(L_{1,R}) \\ \notag &\quad + \P(\{\exists \ell \in \{1,\ldots,k-1\} \ a_\ell = 1 \ \text{or}\ b_\ell = 1 \}).
		\end{align}
		The $N-1$ factor comes from the fact that $\phi(1)$ is equally likely to be each of $\{2,\ldots,N\}$. 	The first probability on the RHS of \eref{Tkr.breakdown} is suitable to apply a transport procedure to move $x_2$ closer to $x_1$, but first we bound the last term. By exchangeability, we have
		\begin{equation} \label{e.albl.bd}
			\P(\{\exists \ell \in \{1,\ldots,k-1\} \ a_\ell = 1 \ \text{or}\ b_\ell = 1 \}) \leq \frac{2k-2}{N}
		\end{equation}
		since $\{a_\ell : \ell=1,\ldots,k-1\} \cup \{a_\ell : \ell=1,\ldots,k-1\}$ is random subset of $\{1,\ldots,N\}$ size at most $2k-2$.
		
		We condition on $X_{N,3} = (x_3,x_4,\ldots,x_N)$ to rewrite
		\begin{align} \label{e.etak.issmall.cond}
			\lefteqn{ \P(\{|x_1 - x_2| \geq r\} \cap \{\phi(1) = 2\} \cap L_{1,R}^c \cap T_{k-1,r}) } \quad & \\ \notag &\leq \P(T'_{k-1,r} \cap \{|x_1 - x_2| \in [r,R]\} \cap \{\phi(1) = 2\}) \\ \notag
			&= \E\left[ \1_{T'_{k-1,r}} \P(\{|x_1 - x_2| \in [r,R]\} \cap \{\phi(1) = 2\} \ | \ X_{N,3}) \right],
		\end{align}
		where $T'_{k-1,r}$ is the event that the cardinality of $\{\{i,j\} \subset \{3,\ldots,N\} \ : \ |x_i - x_j| < r, i \ne j \}$ is at most $k-1$.
		
		We next define a new type of isotropic averaging operator to apply  to the conditional probability above. For a rotationally symmetric probability measure $\nu$ on $\R^d$, define the mimicry operators
		\begin{align*}
			\mathrm{Mim}_{1,2,\nu} F(x_1,x_2) &= \int_{\R^d} F(x_1,x_1 + y)\nu(dy), \\
			\mathrm{Mim}_{2,1,\nu} F(x_1,x_2) &= \int_{\R^d} F(x_2 + y,x_2)\nu(dy).
		\end{align*}
		Note that
		$$
		\mathrm{Mim}_{1,2,\nu} \sum_{j=3}^N \g(x_2 - x_j) \leq \sum_{j=3}^N \g(x_1 - x_j)
		$$
		and
		$$
		\mathrm{Mim}_{1,2,\nu} U_1(x_2) \leq U_1(x_1), \quad
		\mathrm{Mim}_{1,2,\nu} W(x_2) \leq W(x_1) + C r^2
		$$
		since $U_1$ is superharmonic and $\supp \ \nu \subset \ov{B_r(0)}$. Define
		$$
		\Delta_{\nu} = \mathrm{Mim}_{1,2,\nu} \g(x_1 - x_2) = \mathrm{Mim}_{2,1,\nu} \g(x_1 - x_2)
		$$
		which depends only on $\nu$. Analogous results as above hold for $\mathrm{Mim}_{2,1,\nu}$. It follows that
		$$
		e^{-\beta \mathrm{Mim}_{1,2,\nu}  \mcl H^{W, U}(X_N) } + e^{-\beta \mathrm{Mim}_{2,1,\nu}  \mcl H^{W,U}(X_N) } \geq e^{-C\beta r^2} e^{\beta (\g(x_1 - x_2) - \Delta_\nu) } e^{-\beta \mcl H^{W,U}(X_N)}.
		$$
		Indeed, the first summand on the LHS dominates the RHS in the event that
		$$
		\sum_{j=3}^N \g(x_1 - x_j) + U_1(x_1) + W(x_1) \leq \sum_{j=3}^N \g(x_2 - x_j) + U_1(x_2) + W(x_2)
		$$
		since $x_2$ interacts with $x_3,\ldots,x_N$ similar to how $x_1$ does after applying the operator $\mathrm{Mim}_{1,2,\nu}$. When the reverse inequality is true, the second summand on the LHS is dominating. The adjoints are easily computed as
		\begin{align*}
			\mathrm{Mim}_{1,2,\nu}^\ast F(x_1,x_2) &= \nu(x_2 - x_1) \int_{\R^d} F(x_1,y)dy, \\
			\mathrm{Mim}_{1,2,\nu}^\ast F(x_1,x_2) &= \nu(x_1 - x_2) \int_{\R^d} F(y,x_2)dy.
		\end{align*}
		
		We will now apply the model computation to $\P(\cdot \ | \ X_{N,3})$. For a normalizing factor $\mcl Z(X_{N,3})$, we have a.s.\
		\begin{align*}
			\lefteqn{ \P(\{|x_1 - x_2| \in [r,R]\} \cap \{\phi(1) = 2\} \ | \ X_{N,3}) } \quad & \\ &= \frac{1}{\mcl Z(X_{N,3})} \iint_{\R^d \times \R^d} \1_{\{|x_1-x_2| \in [r,R]\} \cap \{\phi(1) = 2\}}(x_1,x_2)e^{-\beta \mcl H^{W,U}(X_N)} dx_1 dx_2 \\ \notag
			&\leq \frac{e^{C\beta r^2 + \beta \Delta_\nu}}{\mcl Z(X_{N,3})} (A_1 + A_2),
		\end{align*}
		where (with a slight abuse of notation)
		$$
		A_1 = \iint_{\R^d \times \R^d} \mathrm{Mim}_{1,2,\nu}^\ast \( e^{-\beta \g(x_1 - x_2)}  \1_{\{|x_1-x_2| \in [r,R]\} \cap \{\phi(1) = 2\}} \)(x_1,x_2) e^{-\beta \mcl H^{W,U}(X_N)} dx_1 dx_2
		$$
		and $A_2$ identical but with $\mathrm{Mim}_{2,1,\nu}^\ast$ in the place of $\mathrm{Mim}_{1,2,\nu}^\ast$. Note that $\mathrm{Mim}_{1,2,\nu}^\ast$ is monotonic and
		$$
		\1_{\{|x_1-x_2| \in [r,R]\} \cap \{\phi(1) = 2\}} \leq \1_{\{|x_1-x_2| \in [r,R]\}},
		$$
		whence
		$$
		\mathrm{Mim}_{1,2,\nu}^\ast \( e^{-\beta \g(x_1 - x_2)}  \1_{\{|x_1-x_2| \in [r,R]\} \cap \{\phi(1) = 2\}} \)(x_1,x_2) \leq \nu(x_2 - x_1)  \int_{\Ann_{[r,R]}(x_1)} e^{-\beta \g(x_1 - y)}  dy.
		$$
		Define $M_{r,R} = \int_{\Ann_{[r,R]}(0)} e^{-\beta \g(y)}dy$. Since $\nu$ is supported in $\ov{B_r(0)}$, the RHS above is a.s.\ bounded by
		$$
		\| \nu \|_{L^\infty(\R^d)} M_{r,R} \1_{B_r(0)}(x_1 - x_2),
		$$
		and we find
		$$
		\frac{A_1}{\mcl Z(X_{N,3})} \leq 	M_{r,R} \| \nu \|_{L^\infty(\R^d)} \P(\{ |x_1 - x_2| < r \} \ | \ X_{N,3}).
		$$
		We can prove an identical bound for $A_2$ with the same argument, and using the bounds in \eref{etak.issmall.cond} shows
		\begin{equation} \label{e.mimic.bd}
			\P(\{|x_1 - x_2| \geq r\} \cap \{\phi(1) = 2\} \cap L_{1,R}^c \cap T_{k-1,r}) \leq \kappa \P( \{|x_1 - x_2| < r\} \cap {T'_{k-1,r}}).
		\end{equation}
		where
		$$
		\kappa := 2e^{C \beta r^2 +  \beta \Delta_\nu}M_{r,R} \| \nu \|_{L^\infty(\R^d)}.
		$$
		We wish to partially recover the information $\phi(1) = 2$ after the transport, i.e.\ on the RHS of \eref{mimic.bd}. To do so, we use the bound
		$$
		\P( \{|x_1 - x_2| < r\} \cap \{\phi(1) \ne 2\}) \leq \frac{C_n}{N} \P(\{X(B_r(x_1)) \geq 3\}) + \P(\{X(B_r(x_1)) \geq n\}),
		$$
		finding (starting from \eref{mimic.bd})
		\begin{align} \label{e.etak.isbig.mimic.bd2}
			\lefteqn{ \P(\{|x_1 - x_2| \geq r\} \cap \{\phi(1) = 2\} \cap L_{1,R}^c \cap T_{k-1,r}) } \quad & \\ \notag &\leq \kappa \bigg (\P( \{|x_1 - x_2| < r\} \cap \{\phi(1) = 2\}\cap {T'_{k-1,r}}) \\ \notag &\quad + \frac{C_n}{N} \P(\{X(B_r(x_1)) \geq 3\}) + \P(\{X(B_r(x_1)) \geq n\}) \bigg).
		\end{align}
		
		Let $T'_{k-1,r}(j)$ be the event that $\{x_i : i \in \{2,\ldots,N\}, i \ne j\}$ contains at most $k-1$ pairs of points within distance $r$. For example, $T'_{k-1,r}(2) = T'_{k-1,r}$. Since the events $\{\phi(i) = j\}$, $j=2,\ldots,N$ are disjoint up to probability $0$, we see
		\begin{align} \label{e.etak.isbig.gain1}
			\lefteqn{ \P( \{|x_1 - x_2| < r\} \cap \{\phi(1) = 2\}\cap {T'_{k-1,r}})} \quad & \\ \notag &= \frac{1}{N-1} \sum_{j=2}^N \P( \{|x_1 - x_j| < r\} \cap \{\phi(1) = j\}\cap {T'_{k-1,r}(j)}) \\ \notag &\leq \frac{1}{N-1}\P(\{|x_1 - x_{\phi(1)}| < r\} \cap T'_{k-1,r}(\phi(1))\}).
		\end{align}
		On the event $\{|x_1 - x_{\phi(1)}| < r\} \cap T'_{k-1,r}(\phi(1))\}$, the index $1$ is an exceptional index since there likely are very few particles within distance $r$ of each other, so we expect the probability of the event is of order $O(1/N)$. To be precise, if $X(B_r(x_i)) \leq n$ for all $i \in \{1,\ldots,N\}$, then $T'_{k-1,j}(\phi(1))$ occurring implies $T_{k-1+2n,r}$. Thus
		\begin{align} \label{e.etak.isbig.gain2}
			\lefteqn{ \P (  {\{|x_1 - x_{\phi(1)}| < r \}} \cap {T'_{k-1,j}(\phi(1))}) } \quad & \\ \notag &\leq \P(\{\exists i \ X(B_r(x_i)) \geq n\}) + \P (  {\{|x_1 - x_{\phi(1)}| < r \}} \cap T_{k-1+2n,r}) \\ \notag &\leq N \P(\{X(B_r(x_1)) \geq n\}) +  \P (  {\{|x_1 - x_{\phi(1)}| < r \}} \cap T_{k-1+2n,r}).
		\end{align}
		Since, by definition of $T_{k-1+2n,r}$, we have pointwise a.s.\
		$$
		\sum_{i=1}^N \1_{\{|x_i - x_{\phi(i)}| < r \}} \1_{T_{k-1+2n, r}}  \leq 2k-2+4n,
		$$
		we can apply exhchangeability to see
		$$
		\P (  {\{|x_1 - x_{\phi(1)}| < r \}} \cap T_{k-1+2n,r}) \leq \frac{2k-2+4n}{N}.
		$$
		Collecting this estimate, \eref{etak.isbig.gain2}, \eref{etak.isbig.gain1}, and \eref{etak.isbig.mimic.bd2}, we have
		\begin{align*}
		\lefteqn{ \P(\{|x_1 - x_2| \geq r\} \cap \{\phi(1) = 2\} \cap L_{1,R}^c \cap T_{k-1,r}) } \quad & \\ &\leq \kappa\(\frac{2k-2+4n}{N(N-1)} + \frac{C_n}{N} \P(\{X(B_r(x_1)) \geq 3\}) + 2\P(\{X(B_r(x_1)) \geq n\})\).
		\end{align*}
		Finally, plugging this bound into \eref{Tkr.breakdown} along with \eref{albl.bd} proves the proposition.
	\end{proof}
	
	\begin{proposition} \label{p.etak.issmall}
		Consider $W = V_N$ with $V$ satisfying \eref{A1} in $d=2$ and \eref{A2} in $d \geq 3$. In $d=2$, the law of $N^{\frac{1}{2+\beta}}\eta_k$ is tight as $N \to \infty$ and $\limsup_{N \to \infty} \P^{V_N}_{N,\beta}(\{N^{\frac{1}{2+\beta}} \eta_k \geq \gamma\}) \leq C\gamma^{-\frac{4+2\beta}{4+\beta}}$ for $\gamma > 0$. For $d \geq 3$, let $Z_k$ be defined by
		$$
		\eta_k = \(\frac{\beta}{\log N} \)^{\frac{1}{d-2}} \( 1 + \frac{2d-2}{(d-2)^2}  \frac{ \log \log N}{\log N} + \frac{Z_k}{(d-2) \log N}  \).
		$$
		Then we have $\limsup_{N \to \infty} \P^{V_N}_{N,\beta}(\{ Z_k \geq \gamma \}) \leq Ce^{-\gamma/2}$. 
	\end{proposition}
	\begin{proof}
		Both results are consequences of \pref{createpairs}, \lref{lonely}, and our clustering result \tref{1C.cluster}. We adopt the notation from \pref{createpairs}. For the $d=2$ result, choose $r = \gamma N^{-\frac{1}{2+\beta}}$ for $\gamma \geq 1$, $n = 5$ (say), and let $\nu$ be the uniform probability measure on $\Ann_{[r/2,r]}(0)$. Without loss of generality, we assume $r \leq 1$. We compute
		\begin{align*}
			&M_{r,R} \leq CR^{2+\beta}, \\
			&\Delta_{\nu} \leq -\log r + C, \\
			&\| \nu \|_{L^\infty(\R^2)} \leq Cr^{-2}.
		\end{align*}
		Applying \pref{createpairs}, \lref{lonely}, and \tref{1C.cluster}, we see
		\begin{equation}
			\P^{V_N}_{N,\beta}(\{\eta_k \geq r\}) \leq CN^{-\frac1d}\sqrt{\log N} + CR^{-2} + C R^{2+\beta} \gamma^{-2-\beta}e^{C\beta N^{-\frac{2}{2+\beta}}\gamma^2}\(1 + N^{1-\frac{4+3\beta}{2+2\beta}} \gamma^{c_\beta}\)
		\end{equation}
		for a constant $C$ depending on $k$ and some constant $c_\beta > 0$. Taking $\limsup_{N \to \infty}$ of both sides and optimizing in $R$ proves that 
		$$
		\limsup_{N \to \infty} \P^{V_N}_{N,\beta}(\{N^{\frac{1}{2+\beta}} \eta_k \geq \gamma\}) \leq CR^{-d} + CR^{2+\beta} \gamma^{-2-\beta} \leq C\gamma^{-\frac{4+2\beta}{4+\beta}}.
		$$
		
		For the $d \geq 3$ result, choose $r = \(\frac{\beta}{\log N - \frac{2d-2}{d-2} \log \log N - \gamma}\)^{\frac{1}{d-2}}$ for $\gamma > 0$ and WLOG assume $r \leq 1$. Let $\nu$ be the uniform measure on the annulus $\Ann_{[(1 -(\log N)^{-1})r, r]}(0)$, and compute
		\begin{align*}
			&M_{r,R} \leq CR^d, \\
			&\Delta_\nu \leq \frac{1}{(1-(\log N)^{-1})^{d-2} r^{d-2}} \leq \frac{1}{r^{d-2}} \(1 + \frac{d-2}{\log N} + C(\log N)^{-2}\), \\
			&\| \nu \|_{L^\infty(\R^d)} \leq \frac{C\log N}{r^d}.
		\end{align*}
		We can then estimate
		\begin{align*}
			\frac{e^{C \beta r^2 + \beta \Delta_{\nu}}}{N} M_{r,R} \| \nu \|_{L^\infty(\R^d)} &\leq CR^d \exp\( \frac{\beta}{r^{d-2}} \(1 + \frac{d-2}{\log N} \) + \log \log N - d \log r - \log N\) \\
			&\leq C R^d \exp\(- \gamma + d-2 - \frac{2d-2}{d-2} \frac{\log \log N}{\log N} \) \leq CR^d e^{-\gamma}.
		\end{align*}
		Thus we have
		\begin{align*}
			\P^{V_N}_{N,\beta}(\{\eta_k \geq r\}) &\leq CN^{-\frac1d} \sqrt{\log N} + CR^{-d} + CR^d e^{-\gamma}(1+N\P^{V_N}_{N,\beta}(\{X(B_r(x_1)) \geq 3\}) \\ &\quad + N^2\P^{V_N}_{N,\beta}(\{X(B_r(x_1)) \geq n\})) \\
			&\leq CN^{-\frac1d} \sqrt{\log N} + CR^{-d} + CR^d e^{-\gamma}(1+Ne^{-\frac{2\beta}{r^{d-2}}} + N^2 e^{-\frac{(n-1) \beta}{r^{d-2}}}) \\
			&\leq CN^{-\frac1d} \sqrt{\log N} + CR^{-d} + CR^d e^{-\gamma}(1+e^{C\gamma - \log N + C\log \log N}),
		\end{align*}
		where we chose $n=4$. We conclude
		$$
		\limsup_{N \to \infty} \P^{V_N}_{N,\beta}(\{\eta_k \geq r\}) \leq CR^{-d} + CR^d e^{-\gamma} \leq Ce^{-\gamma/2}.
		$$
		To conclude the result on $Z_k$, note that
		$$
		r = \(\frac{\beta}{\log N}\)^{\frac{1}{d-2}} \(1 + \frac{2d-2}{(d-2)^2} \frac{\log \log N}{\log N} + \frac{\gamma}{(d-2) \log N}\) + O((\log N)^{-\frac{1}{d-2} - 2 + \ep})
		$$
		for any $\ep > 0$.
	\end{proof}
	
	\section{Discrepancy Bounds} \label{s.discrepancy}
	In this section, we prove \tref{fLL.over} and \tref{fLL.improved}. All implicit constants $C$ may depend on $\beta$ and on various characteristics of $W$ or $V$. We let $\mu(dx) = \frac{1}{c_d} \Delta W(x)dx$ throughout.
	
	It will be necessary to consider more general domains $\Omega$ for our discrepancy bounds. We will work with a domain $\Omega \subset \R^d$ that is an {\it $\alpha$-thin annulus} for fixed $\alpha \in (0,1]$. To be precise, we will take either $\Omega$ to be a ball of radius $R$ if $\alpha = 1$ or $\Omega = \Ann_{[R-\alpha R,R]}(z)$ for $\alpha \in (0,1)$. This means there exists some $C > 0$ such that
	\begin{equation} \label{e.thin.annular}
		C^{-1} \alpha R^{d} \leq |\Omega| \leq C \alpha R^{d} \quad \text{where}\ 2R = \diam(\Omega).
	\end{equation}
	We also define thickened and thinned versions of $\Omega$. For $s \geq 0$, define $\Omega_s = \Omega \cup \{x \in \R^d : \dist(x,\pa \Omega) \leq s\}$, and for $s < 0$, define $\Omega_s = \{x \in \Omega : \dist(x,\pa \Omega) \geq |s|\}$. We will use $2R$ for $\diam(\Omega)$ throughout the first two subsections. We assume $\alpha R \geq C$ dependent on $\beta$, $\sup \Delta W$, and $\inf_{\Omega_R} \Delta W$.

	We will consider the event $E_{\rho,r,M}$ for parameters $\rho,r,M>0$. It is defined by
	\begin{equation} \label{e.Erho.def}
		E_{\rho,r,M} = \{X(\Omega) \geq \mu(\Omega) + \rho |\Omega|\} \cap \{X(\Omega_{5r} \setminus \Omega_{-5r}) \leq M R^{d-1}r \}.
	\end{equation}
	Let $\phi : \R^d \to \R$ be a smooth, nonnegative, radial function with $\int_{\R^d} \phi(x) dx = 1$ and support within $B_1(0)$. For $s > 0$, define $\phi_s = s^{-d} \phi(x/s)$, and $\psi_s = \phi_s \ast \phi_1$. For a measure (or function) $\lambda$, we will write $\Phi_s \lambda$ for $\phi_s \ast \lambda$ when it makes sense. We will apply the isotropic averaging argument with operator $\Iso_{\mbb X(\Omega), \psi_r}$ (identifying $\psi_r$ with the measure $\psi_r dx$). The analysis mainly hinges on a precise lower bound for
	\begin{equation} \label{e.LL.Delta}
		\Delta_{\rho,r,M} := \inf_{X_N \in E_{\rho,r,M}} \mcl H^{W,U}(X_N) - \Iso_{\mbb X(\Omega), \psi_r} \mcl H^{W,U}(X_N).
	\end{equation}
	We want to be able to take $\rho$ to be very small and still achieve $\Delta_{\rho,r,M} \gg 1$ for appropriate choices of $r$ and $M$.
	
	\subsection{Proof idea.}\ \label{s.fLL.idea} The main idea of the proof is to use a continuum approximation for the energy change upon isotropic averaging to find a lower bound for $\Delta_{\rho,r,M}$. To a continuous charge density, we associate an isotropic averaging procedure adapted to $\phi_r$ in which every ``infinitesimal point charge" constituting the continuum is replaced by a infinitesimal charge shaped like $\phi_r$. 
	
	Crucially, the Coulomb energy of a continuous charge distribution contains the self-energy or ``the diagonal". That is, to a bounded charge distribution $\nu$, we associate the Coulomb energy $\frac12\iint \g(x-y) \nu(dx)\nu(dy)$. After isotropic averaging, i.e.\ replacing $\nu$ by $\Phi_r \nu$, the Coulomb energy is 
	$$
	\frac12\iint_{\R^d \times \R^d} \g(x-y) \Phi_r \nu(dx) \Phi_r \nu(dy) = \frac12\iint_{\R^d \times \R^d} (\Phi_r^2 \g)(x-y) \nu(dx)\nu(dy).
	$$
	The continuum Coulomb energy change is therefore
	\begin{equation} \label{e.Erf.def}
		-\mcl E_r(\nu) \quad \text{where}\quad \mcl E_r(\nu) := \frac12 \iint_{\R^d \times \R^d} (\g - \Phi_r^2 \g)(x-y) \nu(dx)\nu(dy).
	\end{equation}
	Two facts form the crux of the method: (1) $\mcl E_r(\nu)$ is convex and (2) $\g - \Phi_r^2 \g$ is compactly supported (on scale $r$).
	
	Considering the energy term $\int W(x) \nu(dx)$ associated to the potential $W$, which one can think of as the Coulomb interaction between $\nu$ and a signed background charge density of $-\mu$, we can precisely compare \eref{Erf.def} to the change in this term upon smearing $\nu$ by $\phi_r$. The energy change is $\iint_{\R^d \times \R^d} (\g - \Phi_r \g)(x-y) \nu(dx)\mu(dy)$. Note that the kernel is different from that of \eref{Erf.def}, but also the $\frac12$ factor is absent in comparison. In the case that $\nu(\Omega) \geq \mu(\Omega) + \rho|\Omega|$ and $\nu$ is supported on $\Omega$, the net energy change is favorable for the choice $r = (\alpha R)^{1/3}$ under some conditions on $\rho > 0$, provided we apply some appropriate modifications near the boundary of $\Omega$ to overcome any boundary layer effects and approximate $\mu$ by a constant measure near $\Omega$. The modifications create only boundary errors since $\g - \Phi_r^2 \g$ is supported on a scale $r$ much smaller than $R$. Our argument is capturing the fact that overcrowding is unfavorable for the interaction on length scale $r$ and below.
	
	The remainder of the proof involves relating our continuum approximation above to the true change in energy of a point configuration upon isotropic averaging. The first step, and the reason for isotropic averaging using $\psi_r$ instead of $\phi_r$, is that we must {\it renormalize} by replacing our point charges by microscopic continuous charges shaped like $\phi_1$. This allows us to make sense of the self energy of the charges (whereas the self-energy of a point charge is infinite) and so directly relate to the continuum problem. As alluded to above, we must also deal with boundary effects, which are well controlled by the parameter $M$ in the event $E_{\rho,r,M}$ and our local law \tref{1C.LL}. Finally, there are some lower order entropy factors to consider, after which our model computation applies to conclude \tref{fLL.over}.
	
	Regarding \tref{fLL.improved}, we use rigidity for the fluctuation of smooth linear statistics from \cite{S22} to find a screening region whenever the absolute discrepancy in a ball $B_R(z)$ is large. A screening region takes the form of an $\alpha$-thin annulus $\Omega$ just inside or outside $\pa B_R(z)$ that has an excess of positive charge. We can then apply our incompressibility estimate to conclude.
	\subsection{Discrepancy upper bound.}\
	In this subsection, we prove \tref{fLL.over}. We will actually generalize the result to $\alpha$-thin annuli $\Omega$ under some conditions.
	
	We will first control $\Delta_{\rho,r,M}$ from \eref{LL.Delta} in terms of an ``energy" functional $\mcl E_r(\nu)$ defined in \eref{Erf.def}. We will only need to consider measures $\nu$ with bounded densities, so for simplicity we restrict to this case from the start.
	
	We let $q \in \R$ denote a positive constant to be chosen later and introduce $\Leb(dx)$ to denote Lebesgue measure on $\R^d$. We let $\nu_{|A}$ denote the restriction of $\nu$ to a set $A$ for any measure $\nu$. Define
	\begin{equation} \label{e.Ernuq.def}
		\mcl E_r(\nu;q) = \mcl E_r(\nu + q\Leb_{|\Omega_{2r} \setminus \Omega}) - q\iint_{\R^d \times \R^d} (\g - \Phi_r \g)(x-y)dx (\nu(dy) + q\Leb_{|\Omega_{r} \setminus \Omega}(dy)).
	\end{equation}
	The above definition models the full energy change for the continuum approximation discussed in \sref{fLL.idea}. Note that we extend the measure $\nu$ slightly past the boundary $\Omega$ by $qdx$ to eliminate boundary layer effects. In \pref{energymin}, we will see that $\nu = q \Leb_{|\Omega}$ is a minimizer of $\mcl E_r(\cdot;q)$ among bounded measures supported on $\Omega$ with mass $q|\Omega|$.

\begin{lemma} \label{l.smooth.iso}
We have
\begin{align*}
	\Iso_{\{1,2\}, \psi_r} \g(x_1 - x_2) &= (\phi_1 \ast \phi_1 \ast \phi_r \ast \phi_r \ast \g)(x_1 - x_2), \\
	\Iso_{\{1,2\}, \psi_r} \g(x_1 - x_3) &\leq \g(x_1 - x_3).
\end{align*}
Furthermore,
\begin{align} \label{e.smooth.iso.pot}
	\Iso_{\{1\},\psi_r} W(x_1) = W(x_1) &+ \int_{\R^d} (\g - \Phi_1 \g)(y - x_1) \mu(dy) \\ \notag &+ \iint_{\R^d \times \R^d} (\g - \Phi_r \g)(y - x) \phi_1(x - x_1) dx \mu(dy).
\end{align}
\end{lemma}
\begin{proof}
The first two equations follow immediately from the definition of $\Iso_{\{1,2\}, \psi_r}$ and superharmonicity of $\g$; see the proof of \pref{1C.LL.iso} for a similar calculation.

Considering now isotropic averaging of the potential term and letting $\sigma$ be surface measure on the unit sphere, we have
\begin{equation}
	\frac{1}{\sigma(B_1(0))} \int_{\pa B_1(0)} W(z + r  \theta) \sigma(d\theta) = W(z) + \int_{B_{r}(z)} (\g(x - z) - \g(r))\mu(dx).
\end{equation}
Defining $\g_s = (\g(x) - \g(s))_+$, we compute using radial symmetry of $\phi_r$:
\begin{align*}
	\phi_r \ast W(z) = \int_{\R^d} \phi_{r}(y) W(z + y) dy &=   \int_{\R^d} \phi_r(y) \frac{1}{\sigma(B_1(0))} \int_{\pa B_1(0)}W(z+ |y| {\theta}) d\theta dy \\ \notag &= W(z) + \(\int_{\R^d} \phi_{r}(y) \g_{|y|} dy\) \ast \mu(x).
\end{align*}
Note that $\g_s = \g - \delta^{(s)} \ast \g$, where $\delta^{(s)}$ is Dirac delta ``smeared" evenly on the sphere $\pa B_s(0)$, and so
$$
\int_{\R^d} \phi_{r}(y) \g_{|y|}(x) dy = \g(x) - \int_{\R^d} \int_{\R^d} \g(x-z) \phi_r(y) \delta^{(|y|)}(dz) dy = \g(x) - \int_{\R^d} \g(x-z) \phi_r(z) dz.
$$
This is $(\g - \Phi_r \g)(x)$. In the last equality, we used that $\int (\delta^{(|y|)}(dz) \phi_r(y)) dy = \phi_r(z)dz$, formally, since $\phi$ is radial.
By applying the above computation twice, we have
\begin{align*}
	\psi_r \ast W(x_i) &= \phi_r \ast (\phi_1 \ast W)(x_i) = (\phi_1 \ast W)(x_i) +  ((\g - \Phi_r \g) \ast \phi_1 \ast \mu)(x_i) \\
	&= W(x_i) + (\g - \Phi_1 \g) \ast \mu(x_i) + ((\g - \Phi_r \g) \ast \phi_1 \ast \mu)(x_i) .
\end{align*}
The last line with $i=1$ gives \eref{smooth.iso.pot}.
\end{proof}

The next lemma directly relates the energy change of a point charge system after isotropic averaging to the functional $\mcl E_r$. There are error terms corresponding to self-energy terms ($\mathrm{Error}_{\mathrm{vol}}$), boundary layer effects ($\mathrm{Error}_{\mathrm{bl}}$), and approximation of $\mu$ by $q$.
\begin{lemma} \label{l.twostep.energy.avg}
Let $r > 1$ and $q \in \R$. We have
\begin{align} \label{e.twostep.energy.avg}
	\lefteqn{\mcl H^{W,U}(X_N) - \Iso_{\mbb X(\Omega), \psi_r} \mcl H^{W,U}(X_N) } \quad & \\ \notag &\geq  \mcl E_r((\Phi_1 X_{| \Omega})_{|\Omega}; q) + \iint_{\R^d \times \Omega} (\g - \Phi_r \g)(y-x) (q \Leb - \mu)(dx) (\Phi_1 X_{| \Omega})(dy) \\ \notag &\quad- \mathrm{Error}_{\mathrm{bl}} - \mathrm{Error}_{\mathrm{vol}},
\end{align}
where
\begin{equation} \label{e.twostep.energy.errorbl}
	\mathrm{Error}_{\mathrm{bl}} \leq Cr^2 |q| X(\Omega \setminus \Omega_{-3r}) + Cr^3 q^2 R^{d-1} + Cr^2 ( 1 + |q|) X(\Omega \setminus \Omega_{-1}),
\end{equation}
and
\begin{equation} \label{e.twostep.energy.errorvol}
	\mathrm{Error}_{\mathrm{vol}} \leq -X(\Omega)\g(4r) + CX(\Omega).
\end{equation}
\end{lemma}
\begin{proof}
By \lref{smooth.iso}, we have
\begin{equation*}
	\Iso_{\mbb X(\Omega),\psi_r} \mcl H^0(X_N) \leq \frac12 \sum_{\substack{i \ne j \\ \{i,j\} \not \subset \mbb X(\Omega)}} \g(x_i - x_j) + \frac12 \sum_{i,j \in \mbb X(\Omega)} (\Phi_1^2 \Phi_r^2 \g)(x_i - x_j) - \frac12 X(\Omega) \Phi_1^2 \Phi_r^2 \g(0).
\end{equation*}
Note that we have added and subtracted the diagonal terms in the second sum on the RHS. We use that $\Phi_1$ is self-dual to write
$$
\frac12 \sum_{i,j \in \mbb X(\Omega)} (\Phi_1^2 \Phi_r^2 \g)(x_i - x_j) = \frac12 \iint_{\R^d \times \R^d} \Phi_r^2 \g(x - y) (\Phi_1 X_{|\Omega})(dx)(\Phi_1 X_{|\Omega})(dy).
$$
Considering the interaction within $\mcl H^{0}(X_N)$, we estimate
$$
\frac12 \sum_{\substack{i,j \in \mbb X(\Omega) \\ i \ne j}} \g(x_i - x_j) \geq \frac12 \iint_{\R^d \times \R^d} \g(x - y) (\Phi_1 X_{|\Omega})(dx)(\Phi_1 X_{|\Omega})(dy) - \frac12 X(\Omega) \Phi_1^2 \g(0)
$$
using $\Phi_1^2 \g \leq \g$ pointwise and duality like above. It follows that
\begin{align}
	\mcl H^0(X_N) - \Iso_{\mbb X(\Omega),\psi_r} \mcl H^0(X_N) &\geq  \frac12 \iint_{\R^d \times \R^d} (\g - \Phi_r^2 \g)(x - y) (\Phi_1 X_{|\Omega})(dx)(\Phi_1 X_{|\Omega})(dy) \\ \notag &\quad + \frac12 X(\Omega) (\Phi_1^2 \Phi_r^2 \g (0)-\Phi_1^2 \g(0)).
\end{align}
We can compare the double integral above to $\mcl E_r((\Phi_1 X_{|\Omega})_{|\Omega} + q\Leb_{|\Omega_{2r} \setminus \Omega})$. They will only differ by boundary layer terms. We start by simply restricting the integral to $\Omega \times \Omega$ using $\g - \Phi_r^2 \g \geq 0$. After doing so, we find
$$
\frac12 \iint_{\R^d \times \R^d} (\g - \Phi_r^2 \g)(x - y) (\Phi_1 X_{|\Omega})^{\otimes 2} (dx,dy) \geq \mcl E_r((\Phi_1 X_{|\Omega})_{|\Omega} + q \Leb_{|\Omega_{2r} \setminus \Omega}) - T_1 - T_2
$$
where
\begin{align*}
	T_1 &= q\iint_{\Omega \times \R^d} (\g - \Phi_r^2 \g)(x-y) \Phi_1 X_{|\Omega}(dx) \Leb_{|\Omega_{2r} \setminus \Omega}(dy) \leq Cr^2 |q| X(\Omega \setminus \Omega_{-2r-1}), \\
	T_2 &=q^2\iint_{\R^d \times \R^d} (\g - \Phi_r^2 \g)(x-y) \Leb_{|\Omega_{2r} \setminus \Omega}(dy) \Leb_{|\Omega_{2r} \setminus \Omega}(dy) \leq C r^3 q^2 R^{d-1}.
\end{align*}
We also bound
$$
\frac12 X(\Omega) (\Phi_1^2 \Phi_r^2 \g (0)-\Phi_1^2 \g(0)) \geq  X(\Omega) (\g(4r) - C).
$$
This term, as well as $T_1$ and $T_2$, are absorbed into $\mathrm{Error}_{\mathrm{vol}}$ and $\mathrm{Error}_{\mathrm{bl}}$.

We now handle the potential terms. Like usual, we can handle the $U$ term using superharmonicity. For the $W$ term, we use \lref{smooth.iso} to see
\begin{align} \label{e.W.iso.1}
	&\Iso_{\mbb X(\Omega),\psi_r} \sum_{i \in \mbb X(\Omega)} W(x_i) = \sum_{i \in \mbb X(\Omega)} W(x_i) \\ \notag &\quad + \iint_{\R^d \times \R^d} (\g - \Phi_1 \g)(x-y) X_{|\Omega}(dx) \mu(dy) + \iint_{\R^d \times \R^d} (\g - \Phi_r \g)(x-y) (\Phi_1 X_{|\Omega})(dx) \mu(dy).
\end{align}
The middle term on the RHS in \eref{W.iso.1} contributes to the volume error $\mathrm{Error}_{\mathrm{vol}}$. It is bounded by
$$
\iint_{\R^d \times \R^d} (\g - \Phi_1 \g)(x-y) X_{|\Omega}(dx) \mu(dy) \leq C X(\Omega).
$$
For the last term in \eref{W.iso.1}, we replace $\mu$ by $q$ to generate the term
$$
\iint_{\R^d \times \R^d} (\g - \Phi_r \g)(x-y) (\Phi_1 X_{|\Omega})(dx) (q\Leb - \mu)(dy)
$$
in \eref{twostep.energy.avg}. We then estimate
\begin{align*}
	\lefteqn{ q\iint_{\R^d \times \R^d} (\g - \Phi_r \g)(x-y) (\Phi_1 X_{|\Omega})(dx) dy } \quad & \\ &\leq q\iint_{\Omega \times \R^d} (\g - \Phi_r \g)(x-y) (\Phi_1 X_{|\Omega})(dx) dy + Cr^2 |q| X(\Omega \setminus \Omega_{-1})
	\\ &\leq q\iint_{\R^d \times \R^d} (\g - \Phi_r \g)(x-y) ((\Phi_1 X_{|\Omega})_{|\Omega} +  q\Leb_{|\Omega_r \setminus \Omega})(dx) dy + Cr^2 |q| X(\Omega \setminus \Omega_{-1}),
\end{align*}
where we used $\| \g - \Phi_r \g \|_{L^1(\R^d)} \leq Cr^2$ and the fact that the mass of $\Phi_1 X_{|\Omega}$ lying outside of $\Omega$ is bounded by $X(\Omega \setminus \Omega_{-1})$. Assembling the above estimates proves the lemma.
\end{proof}	

We now turn to studying minimizers of $\mcl E_r(\cdot;q)$ conditioned on the weight given to $\Omega$. The following proposition is the key technical result of \sref{discrepancy}, and it is the reason for considering the precise form of the energy $\mcl E_r(\cdot;q)$.
\begin{proposition} \label{p.energymin}
Let $\Omega$ an $\alpha$-thin annulus and $q \in \R$. We have
\begin{equation} \label{e.energymin.inf}
	\inf_{\nu : \nu(\Omega)= q |\Omega|} \mcl E_r(\nu;q) \geq 0
\end{equation}
where the infimum is over measures $\nu$ supported on $\ov \Omega$ with a bounded Lebesgue density.
\end{proposition}
\begin{proof}
First, note that $\mcl E_r(\cdot)$ as defined in \eref{Erf.def} is non-negative when applied to measures with bounded density and compact support. This is because the Fourier transform of $\delta_0 - \phi_r \ast \phi_r$ is non-negative. Indeed, the Fourier transform on $\R^d$ of $\phi_r \ast \phi_r$ is real since $\phi_r$ is radial, and it is bounded above by $1$ (in the normalization for which $\hat{\delta}_0 = 1$) since it is a probability density. Since $\g$ has positive Fourier transform, the Fourier transform of $\g - \Phi_r^2 \g = (\delta_0 - \phi_r \ast \phi_r) \ast \g$ is non-negative, and non-negativity of $\mcl E_r(\cdot)$ follows from Plancherel's theorem.

Let $\nu$ be a measure supported on $\ov{\Omega}$ with a bounded Lebesgue density and $\nu(\Omega) = q|\Omega|$. We use that $\mcl E_r(\cdot;q)$ is quadratic to expand
\begin{align} \label{e.Er.quadratic}
	\mcl E_r(\nu;q) &= \mcl E_r(q \Leb_{|\Omega};q) + \mcl E_r(q \Leb_{|\Omega}-\nu) \\ \notag &\quad + \iint_{\R^d \times \R^d} (\g - \Phi_r^2 \g)(x-y) (\nu - q \Leb_{|\Omega})(dx) (q \Leb_{|\Omega_{2r}})(dy) \\ \notag &\quad - q\iint_{\R^d \times \R^d} (\g - \Phi_r \g)(x-y) dx (\nu - q\Leb_{|\Omega})(dy).
\end{align}
We claim that both terms on the last line are $0$. Indeed, we have
\begin{equation} \label{e.qconv}
	\int_{\R^d} (\g - \Phi_r \g)(x-y) dy =  c_{1,\phi,r},\quad \int_{\R^d} (\g - \Phi_r^2 \g)(x-y) dy =  c_{2,\phi,r}
\end{equation}
for constants $c_{1,\phi,r}, c_{2,\phi,r}$ independent of $x$, and the same holds for $\Leb_{|\Omega_{2r}}(dy)$ in place of $dy$ as long as $x \in \Omega$. The claim follows from $\nu(\Omega) = q$. Since $\mcl E_r(q\Leb_{|\Omega}-\nu) \geq 0$, we have proved that the infimum in \eref{energymin.inf} is attained at $\nu = q \Leb_{|\Omega}$.

It remains to compute $\mcl E_r(q \Leb_{|\Omega};q)$. We write $\mcl E_r(q \Leb_{|\Omega};q)= T_1 + T_2 - T_3$ for
\begin{align*}
	T_1 &= \frac{q^2}{2}\iint_{\R^d \times \R^d} (\g - \Phi_r \g)(x-y)\Leb_{|\Omega_{2r}}(dx)\Leb_{|\Omega_{2r}}(dy), \\ T_2 &= \frac{q^2}{2}\iint_{\R^d \times \R^d} \Phi_r (\g - \Phi_r \g)(x-y) \Leb_{|\Omega_{2r}}(dx)\Leb_{|\Omega_{2r}}(dy), \\ T_3 &= q^2\iint_{\R^d \times \R^d} (\g - \Phi_r \g)(x-y) dx \Leb_{|\Omega_r}(dy).
\end{align*}
Note that $\g - \Phi_r \g \geq 0$, and so
$$
T_1 \geq \frac{q^2}{2} |\Omega_r| \inf_{y \in \Omega_r} \int_{\Omega_{2r}}(\g - \Phi_r \g)(x-y) = \frac12 q^2 c_{1,\phi,r} |\Omega_r|.
$$
We use $\Phi_r \Leb_{|\Omega_{2r}} \geq \Leb_{|\Omega_r}$ to see
\begin{align*}
T_2 &= \frac{q^2}{2} \iint_{\R^d \times \R^d} (\g - \Phi_r \g)(x-y) \Phi_r \Leb_{|\Omega_{2r}}(dx) \Leb_{|\Omega_{2r}}(dy) \\ &\geq \frac{q^2}{2} \iint_{\Omega_r \times \Omega_{2r}} (\g - \Phi_r \g)(x-y) dx dy \geq \frac{c_{1,\phi,r} q^2}{2} |\Omega_r|.
\end{align*}
Similarly, we have $T_3 \leq q^2 c_{1,\phi,r} |\Omega_r|$, and combining the bounds on $T_1,T_2,T_3$ finishes the proof.
\end{proof}

\begin{proposition} \label{p.fLL.Delta}
Let $\Omega$ be an $\alpha$-thin annulus with diameter $2R$. Assume that the parameters $\rho,\alpha,q > 0$ and $M \geq 1$ satisfy the following for fixed constants $C_i$, $i=1,2,3,4,5$.
\begin{enumerate}
	\item There is bounded excess:
	\begin{equation} \label{e.rhobdd}
		\rho \leq C_1.
	\end{equation}
	\item The constant $q$ approximates $\mu$:
	\begin{equation} \label{e.qapproxmu}
		\| \mu - q \|_{L^\infty(\Omega_R)} \leq C_2^{-1} \rho.
	\end{equation}
	\item The annulus is not too thin:
	\begin{equation} \label{e.thickannulus}
		\alpha R \geq C_3.
	\end{equation}
	\item There is significant charge excess:
	\begin{equation} \label{e.bigcharge}
		\frac{(\alpha R)^{2/3} \rho}{1 + \1_{d=2}\log (\alpha R)} \geq C_4.
	\end{equation}
	\item The boundary layer density is not too high:
	\begin{equation} \label{e.lowbdrydensity}
		M \leq C_5^{-1} (\alpha R)^{2/3}\rho.
	\end{equation}
\end{enumerate}
Assume that, dependent on $(\inf_{\Omega} \mu)^{-1}$, $\sup_{\Omega_{R}} \mu$, $\beta$, we have $C_i \gg 1$ for $i =2,3,4,5$ and $C_2 \geq C_1$. Then we have
\begin{equation}
	\Delta_{\rho,r,M}  \geq C^{-1} (\alpha R)^{2/3} \rho ( \mu(\Omega) + \rho |\Omega|)
\end{equation}
for $r = (\alpha R)^{1/3}$ and the quantity $\Delta_{\rho,r,M}$ defined in \eref{LL.Delta}.
\end{proposition}
\begin{proof}
Let $r = (\alpha R)^{1/3}$ and let $X_N \in E_{\rho,r,M}$ be arbitrary.

{\it Step 1:} We begin by estimating the energy change for the continuum approximation discussed in \sref{fLL.idea}. First, note that $C_2 \geq C_1$ means that
$$
q \leq \| \mu - q \|_{L^\infty(\Omega_R)} + \| \mu \|_{L^\infty(\Omega_R)} \leq C_2^{-1} \rho + C \leq C_2^{-1} C_1 + C \leq C.
$$
Define $q_X = q + m_X$ for
$$
m_X := \frac{1}{|\Omega|}(\Phi_1 X_{|\Omega} (\Omega) - q |\Omega|).
$$
The parameter $m_X$ acts as an excess charge density beyond $q$ accounting for boundary layer effects. We estimate it by
\begin{align*}
	\Phi_1 X_{|\Omega}(\Omega) \geq X(\Omega) - X(\Omega \setminus \Omega_{-1}) &\geq \mu(\Omega) + \rho |\Omega| - MR^{d-1}r \\ \notag &\geq (q + \rho) |\Omega| - MR^{d-1}r - \| \mu - q \|_{L^\infty(\Omega)} |\Omega|.
\end{align*}
Applying our assumptions, we see
\begin{align} \label{e.mX.lower}
m_X \geq \rho - \frac{MR^{d-1}r}{|\Omega|} -\| \mu - q \|_{L^\infty(\Omega)} &\geq \rho - CM(\alpha R)^{-2/3} - \| \mu - q \|_{L^\infty(\Omega)} \\ \notag &\geq \rho - C C_5^{-1} \rho - C_2^{-1} \rho \geq \frac12 \rho.
\end{align}
Since $q_X|\Omega| = \Phi_1 X_{|\Omega}(\Omega)$, by \pref{energymin} we have
$$
\mcl E_r((\Phi_1 X_{|\Omega})_{|\Omega}; q_X) \geq 0.
$$
We now lower bound $\mcl E_r((\Phi_1 X_{|\Omega})_{|\Omega}; q)$ by comparison to $\mcl E_r((\Phi_1 X_{|\Omega})_{|\Omega}; q_X)$. First, for any measure $\nu$ with bounded density, we use that $\mcl E_r(\cdot)$ is quadratic to compute
\begin{align*}
	\lefteqn{ \mcl E_r(\nu + q\Leb_{|\Omega_{2r}\setminus \Omega}) - \mcl E_r(\nu + q_X \Leb_{|\Omega_{2r} \setminus \Omega}) } \quad & \\ &= \frac{q-q_X}{2} \iint_{\R^d \times (\Omega_{2r}\setminus \Omega)} (\g - \Phi_r^2 \g)(x-y) (2\nu + (q_X + q)\Leb_{|\Omega_{2r} \setminus \Omega})(dx)dy \\
	&\geq -\frac{|q-q_X|}{2} \| \g - \Phi_r^2 \g \|_{L^1(\R^d)} (2|\nu| + (q_X+q)\Leb_{|\Omega_{2r} \setminus \Omega})(\Omega_{5r} \setminus \Omega_{-3r}) \\ &\geq -C|q-q_X| r^2 (|\nu|(\Omega_{5r}\setminus \Omega_{-2r}) + CR^{d-1} r).
\end{align*}
The bound in the second to last line follows from the fact that $\g - \Phi_r^2 \g$ has support of diameter at most $4r$ and so we only need to integrate over $x \in \Omega_{5r} \setminus \Omega_{-3r}$. For $\nu = (\Phi_1 X_{|\Omega})_{|\Omega}$, we see
$$
\mcl E_r((\Phi_1 X_{|\Omega})_{|\Omega} + q\Leb_{|\Omega_{2r}\setminus \Omega}) - \mcl E_r((\Phi_1 X_{|\Omega})_{|\Omega} + q_X \Leb_{|\Omega_{2r} \setminus \Omega}) \geq -Cm_X (M+C)R^{d-1}r^3.
$$
Plugging this computation into \eref{Ernuq.def}, we find
\begin{align} \notag
	\lefteqn{ \mcl E_r((\Phi_1 X_{|\Omega})_{|\Omega}; q) - \mcl E_r((\Phi_1 X_{|\Omega})_{|\Omega}; q_X) } \quad & \\ \notag &\geq -Cm_X (M+C)R^{d-1}r^3 + q_X \iint_{\R^d \times \R^d} (\g - \Phi_r \g)(x-y) dx (\Phi_1 X_{|\Omega} + q_X \Leb_{|\Omega_r \setminus \Omega})(dy) \\ \notag &\quad -  q \iint_{\R^d \times \R^d} (\g - \Phi_r \g)(x-y) dx (\Phi_1 X_{|\Omega} + q \Leb_{|\Omega_r \setminus \Omega})(dy) \\ \notag &\geq -Cm_X (M+C)R^{d-1}r^3 + (q_X - q) \iint_{\R^d \times \R^d} (\g - \Phi_r \g)(x-y) dx (\Phi_1 X_{|\Omega} + q \Leb_{|\Omega_r \setminus \Omega})(dy) \\ \notag \label{e.Ercomp.final}
	&\geq -Cm_X(M+C)\alpha R^{d} + m_X \| \g - \Phi_r \g \|_{L^1(\R^d)} \Phi_1 X_{|\Omega}(\Omega).
\end{align}
We can compute $\| \g - \Phi_r \g \|_{L^1(\R^d)}  = c_{\phi} r^2$ for some constant $c_\phi > 0$. Furthermore by \eref{rhobdd}, \eref{thickannulus}, and \eref{lowbdrydensity}, we have $(M+C)\alpha R^d \ll r^2 \alpha R^d \leq C c_\phi r^2 \Phi_1 X_{|\Omega}(\Omega)$. It follows
$$
 \mcl E_r((\Phi_1 X_{|\Omega})_{|\Omega}; q) - \mcl E_r((\Phi_1 X_{|\Omega})_{|\Omega}; q_X) \geq \frac12 c_\phi m_X r^2 \Phi_1 X_{|\Omega}(\Omega).
$$
We apply our estimate \eref{mX.lower} on $m_X$ and the similar estimate $\Phi_1 X_{|\Omega}(\Omega) \geq \frac12 X(\Omega)$ to see 
\begin{equation} \label{e.Ergain.lwrbd0}
	\mcl E_r((\Phi_1 X_{|\Omega})_{|\Omega}; q) \geq \frac14 c_\phi r^2 \rho X(\Omega).
\end{equation}
This term will be the main energy benefit of isotropic averaging. Using $\inf_\Omega \mu \geq C^{-1}$, we can estimate
\begin{equation} \label{e.Ergain.lwrbd}
	\mcl E_r((\Phi_1 X_{|\Omega})_{|\Omega}; q) \geq C^{-1} \rho r^2 \alpha R^d \geq C^{-1} C_4 \alpha R^d.
\end{equation}

{\it Step 2:} We now relate the isotropic averaging energy change associated to $X$ to the quantity $\mcl E_r((\Phi_1 X_{|\Omega})_{|\Omega}; q)$ using \lref{smooth.iso}. We claim that
\begin{equation} \label{e.fLL.step2claim}
	\mcl H^{W,U}(X_N) - \Iso_{\mbb X(\Omega), \psi_r} \mcl H^{W,U}(X_N) \geq \frac{1}{2} \mcl E_r((\Phi_1 X_{|\Omega})_{|\Omega}; q).
\end{equation}
This claim, combined with \eref{Ergain.lwrbd0}, will finish the proof.
By step 1, it is sufficient to bound the error terms within \lref{twostep.energy.avg} by $\frac18 c_{\phi}r^2 \rho X(\Omega)$. Using the notation from that lemma, we compute
\begin{align*}
	\mathrm{Error}_{\mathrm{bl}} &\leq Cr^3 (1+q) MR^{d-1} + Cr^3 q^2 R^{d-1} \leq C \alpha R^d (M+1) \\ &\leq C \alpha R^d(C_5^{-1} (\alpha R)^{2/3} \rho + 1) \leq C \alpha R^d + C C_5^{-1} \rho r^2 X(\Omega).
\end{align*}
Applying \eref{Ergain.lwrbd0} and \eref{Ergain.lwrbd}, we see that
$$
\mathrm{Error}_{\mathrm{bl}} \leq \frac{1}{100} \mcl E_r((\Phi_1 X_{|\Omega})_{|\Omega}; q).
$$
Next, we estimate the volume error term. In $d=2$, we can compute using \eref{bigcharge}
\begin{equation}
	\mathrm{Error}_{\mathrm{vol}} \leq X(\Omega) (\log(\alpha R) + C) \leq C_4^{-1}X(\Omega)(\alpha R)^{2/3}\rho + CX(\Omega) \leq \frac{1}{400} c_\phi r^2 \rho X(\Omega).
\end{equation}
In the last inequality, we used that $r^2 \rho \geq C_4 \gg 1$. In $d \geq 3$, we can delete $\log(\alpha R)$ above and have the same final result. By \eref{Ergain.lwrbd0}, we have dominated the volume error term by $\frac{1}{100}\mcl E_r((\Phi_1 X_{|\Omega})_{|\Omega}; q)$. The last remaining error term is related to the approximation of $\mu$ by $q$:
\begin{equation}
	\iint_{\R^d \times \Omega}(\g - \Phi_r \g)(q-\mu)(dx)(\Phi_1 X_{|\Omega})(dy) \geq -c_\phi r^2 \|q - \mu\|_{L^\infty(\Omega_r)} X(\Omega) \geq - C_2^{-1} c_\phi r^2 \rho X(\Omega).
\end{equation}
Using \eref{Ergain.lwrbd0} and $C_2 \gg 1$ allows us to dominate this term as well. Assembling  the above allows us to conclude the claim \eref{fLL.step2claim} and the proof.
\end{proof}

We will apply \pref{fLL.Delta} to prove \tref{fLL.over}, but first we take care of the case in which $\rho$ is very large using our high density law \tref{1C.LL}.
\begin{proposition} \label{p.fLL.highdensity}
Let $\Omega$ be an $\alpha$-thin annulus of diameter $2R$ such that $\alpha R \geq 1 - \log \alpha$. Let $\rho \geq C_1$ for a fixed constant $C_1$ taken large enough. Then we have
$$
\P^{W,U}_{N,\beta} (\{X(\Omega) \geq \mu(\Omega) + \rho |\Omega| \}) \leq e^{-\alpha^{d+2} R^{d+2}}.
$$
\end{proposition}
\begin{proof}
Let $(A_\lambda)_{\lambda \in \Lambda}$ be a covering of $\Omega$ by balls $A_\lambda$ of radius $\alpha R$ of cardinality at most $C\alpha^{-d+1}$. We have
\begin{align*}
	\P^{W,U}_{N,\beta} (\{X(\Omega) \geq \mu(\Omega) + \rho |\Omega| \}) &\leq \P^{W,U}_{N,\beta} (\{X(\Omega) \geq \rho |\Omega| \}) \\ &\leq |\Lambda| \sup_{\lambda \in \Lambda} \P^{W,U}_{N,\beta}(\{X(A_\lambda) \geq C^{-1} \alpha^{d-1} \rho |\Omega| \}).
\end{align*}
If $C_1$ is large enough, we have that $\alpha^{d-1} \rho |\Omega| \gg |A_\lambda|$, so we may apply \tref{1C.LL} to see
\begin{align*}
\P^{W,U}_{N,\beta} (\{X(\Omega) \geq \mu(\Omega) + \rho |\Omega| \})  &\leq C \alpha^{-d+1} e^{-C^{-1} (\alpha \rho |\Omega|)^{d+2}} \\ &\leq C \alpha^{-1} e^{-C^{-1} C_1^2\alpha^{d+2}R^{d+2}} \leq  e^{-\alpha^{d+2} R^{d+2}}.
\end{align*}
\end{proof}

We now consider the more general case.
\begin{proposition} \label{p.fLL.upbd}
Let $\Omega$ be an $\alpha$-thin annulus of diameter $2R$. Let $C_1$ be the constant fixed in \pref{fLL.highdensity}. Suppose we have parameters $q > 0$ and $\rho > 0$ such that conditions \eref{qapproxmu}, \eref{thickannulus}, \eref{bigcharge} hold for large enough constants $C_2, C_3, C_4$ with $C_i \gg C_1$ for $i=2,3,4$. Then we have
\begin{equation}
	\P^{W,U}_{N,\beta} (\{X(\Omega) \geq \mu(\Omega) + \rho |\Omega| \}) \leq e^{- c(\alpha R)^{2/3} \rho(\mu(\Omega) + \rho|\Omega|)} +  e^{-c(\alpha R)^{d/3 + 2} \rho^2} + e^{-\alpha^{d+2} R^{d+2}},
\end{equation}
for some $c > 0$.
\end{proposition}
\begin{proof}
Let $r = (\alpha R)^{1/3}$ and $M = C_5^{-1} \rho r^2$ and for a large enough constant $C_5$ ($C_5 = \sqrt{C_4}$ will work if $C_4$ is chosen large enough). 

{\it Step 1:} We first bound the probability of $E_{\rho,r,M}$ by the isotropic averaging argument. Let $n_\rho = \lceil \rho |\Omega| + \mu(\Omega) \rceil$ and $M_\rho = \lceil M R^{d-1} r\rceil$ and $N_1 = \lfloor C_1 |\Omega| + \mu(\Omega) \rfloor$. We write
\begin{equation}
	E_{\rho,r,M} \subset \{X(\Omega) \geq \mu(\Omega) + C_1|\Omega|\} \cup \bigcup_{m = 0}^{M_\rho} \bigcup_{n = n_\rho}^{N_1} \bigcup_{\substack{\M \subset \{1,\ldots,N\} \\ |\M| = m}} \bigcup_{\substack{\N \subset \{1,\ldots,N\} \\ |\N| = n}} E_{\M, \N, r, M}
\end{equation}
for
$$
E_{\M, \N, r, M} = \{\mbb X(\Omega) = \N \} \cap \{\mbb X(\Omega_{5r} \setminus \Omega) = \M \} \cap \{ X(\Omega_{5r} \setminus \Omega_{-5r} ) \leq M R^{d-1}r \}.
$$
Note that
$$
\Iso_{\N, \psi_r}^\ast \1_{E_{\M, \N, r, M}} \leq \1_{\{\mbb X(\Omega_{5r}) = \M\cup \N\}}.
$$
Indeed, since $\Iso^\ast_{\N, \psi_r}$ is convolution by a probability measure, it is bounded by $1$ as an operator $L^\infty \to L^\infty$. By a similar argument as in the proof of \pref{1C.LL.iso}, we have $\Iso_{\N, \psi_r}^\ast \1_{E_{\M, \N, r, M}}(X_N) = 0$ if $\mbb X(\Omega_{5r}) \ne \M\cup \N$.

By isotropic averaging, we conclude, for $|\N| = n \leq N_1$, that
$$
\P^{W,U}_{N,\beta}(E_{\M, \N, r, M}) \leq e^{-\beta \Delta_{\rho_n, r, M}}\P^{W,U}_{N,\beta}(\{\mbb X(\Omega_{5r}) = \M\cup \N\})
$$
where $\rho_n := \frac{n - \mu(\Omega)}{|\Omega|}$ and $\Delta_{\rho_n, r, M}$ is as in \eref{LL.Delta}. By \pref{fLL.highdensity}, a union bound, and exchangeability, we have
\begin{align*}
\P^{W,U}_{N,\beta}(E_{\rho,r,M}) &\leq e^{-\alpha^{d+2} R^{d+2}} \\ &\quad + \sum_{m = 0}^{M_\rho} \sum_{n = n_\rho}^{N_1} {\binom{N}{m}}{\binom{N-m}{n}} e^{-\beta \Delta_{\rho_n, r, M}} \P^{W,U}_{N,\beta}(\{\mbb X(\Omega_{5r}) = \{1,\ldots,m+n\}\}).
\end{align*}
Above, we only summed over $\M$ and $\N$ disjoint. We have
$$
\P(\{\mbb X(\Omega_{5r}) = \{1,\ldots,m+n\}\}) = \frac{1}{{\binom{N}{m+n}}}\P(\{ X(\Omega_{5r}) = m+n\}\}) 
$$
and
$$
\frac{{\binom{N}{m}}{\binom{N-m}{n}}}{{\binom{N}{m+n}}} = \frac{(m+n)!}{m! n!} \leq \frac{(n+m)^m}{m!} \leq 2n^m
$$
whenever $m \leq n$, which for us is always the case if $C_5$ is large enough (independently of $C_i$, $i=2,3,4$). Letting $j = m+n$, we have
\begin{align} \label{e.fLL.prob.E.bd}
	\P^{W,U}_{N,\beta}(E_{\rho,r,M}) &\leq  e^{-\alpha^{d+2} R^{d+2}} + \sum_{n = n_\rho}^{N_1} \sum_{j=n}^{n+M_\rho} 2n^{j - n}e^{-\beta \Delta_{\rho_n, r, M}} \P^{W,U}_{N,\beta}(\{X(\Omega_{5r}) = j\}) \\ \notag &\leq  e^{-\alpha^{d+2} R^{d+2}} + \sum_{n=n_\rho}^{N_1} 2n^{M_\rho}e^{-\beta \Delta_{\rho_n,r,M}}.
\end{align}
Note that $M_\rho \leq 2MR^{d-1} r = 2C_5^{-1} R^{d-1} r^3 \rho \leq C_5^{-1} \alpha R^d r^2 \rho / \log(N_1)$. Thus for $n \leq N_1$ we have
$$
n^{M_\rho} \leq e^{M_\rho \log N_1} \leq  e^{CC_5^{-1}r^2 \rho |\Omega|}.
$$
Since $\beta \Delta_{\rho_n,r,M} \geq C^{-1} r^2 \rho  |\Omega|$ by \pref{fLL.Delta}, we can bound
$$
n^{M_\rho} e^{-\beta \Delta_{\rho_n,r,M}} \leq e^{-\frac{\beta}{2} c r^2 \rho_n (\mu(\Omega) + \rho_n |\Omega|)}
$$
for some $c > 0$ for $C_5$ large enough. The series in the RHS of \eref{fLL.prob.E.bd} can therefore be dominated by a geometric series with rate $1/2$ with the same first term, and so
\begin{equation} \label{e.fLL.prob.E.bd.final}
	\P^{W,U}_{N,\beta}(E_{\rho,r,M}) \leq  e^{-\alpha^{d+2} R^{d+2}} + 4e^{-\frac{\beta}{2} c r^2 \rho(\mu(\Omega) + \rho|\Omega|)}.
\end{equation}

{\it Step 2:} We now consider the event that $X(\Omega) \geq \mu(\Omega) + \rho |\Omega|$ on the complement of $E_{\rho,r,M}$. It suffices to bound the probability that
$
X(\Omega_{5r} \setminus \Omega_{-5r}) \geq M R^{d-1} r.
$
Let $\{  A_\lambda \}_{\lambda \in  \Lambda}$ be a covering of $\Omega_{5r}\setminus \Omega_{-5r}$ by balls of radius $r$ of cardinality at most $C (R/r)^{d-1}$. We have
$$
\P^{W,U}_{N,\beta} (\{ X(\Omega_{5r} \setminus \Omega_{-5r}) \geq M R^{d-1} r \}) \leq \(\frac{CR}{r}\)^{d-1}\sup_{\lambda \in \Lambda} \P^{W,U}_{N,\beta} (\{ X( A_\lambda) \geq C^{-1} Mr^d \}).
$$
Note that $(\alpha R)^{2/3} \rho \geq C_4$ by \eref{bigcharge}, so we have
$
M \geq C_5^{-1} C_4.
$
If we take $C_4 \gg C_5$, we have $M \gg 1$ and we may apply \tref{1C.LL} to see
$$
\P^{W,U}_{N,\beta} (\{ X(\Omega_{5r} \setminus \Omega_{-5r}) \geq M R r \})  \leq \(\frac{CR}{r}\)^{d-1} e^{-C^{-1} r^{d+6} \rho^2} \leq e^{-c r^{d+6} \rho^2}
$$
for some $c > 0$. Assembling \eref{fLL.prob.E.bd.final} and the above finishes the proof.
\end{proof}

It is now a simple matter to prove \tref{fLL.over}.
\begin{proof}[Proof of \tref{fLL.over}]
We will apply \pref{fLL.upbd} to the $1$-thin annulus $\Omega = B_R(z)$ with a certain parameter $\rho$. For the constant $q$ approximation to $\mu = \frac{1}{c_d} \Delta W$, we take $q = \mu(z)$. If $W$ is quadratic near $B_R(z)$, this approximation is exact, but we focus on the general case. One finds that
$$
\| \mu - q \|_{L^\infty(B_{2R}(z))} \leq CN^{-1/d} R,
$$
whence we have the restriction $\rho \gg N^{-1/d}R$ in applying \pref{fLL.upbd}. We also have the restriction $\rho \gg R^{-2/3} (1+\1_{d=2}\log R)$ from \eref{bigcharge}. If
$
R \leq N^{\frac{3}{5d}},
$
then the latter restriction is the only relevant one, and we achieve
$$
\P^{W,U}_{N,\beta}(\{X(B_R(z)) \geq \mu(B_R(z)) + \rho |B_R(z)|\}) \leq e^{-c R^d T} + e^{-c R^{(d+2)/3} T^2} + e^{- R^{d+2}},
$$
as desired, where we set $\rho = TR^{-2/3} (1+\1_{d=2}\log R)$ for a large $T > 0$.
\end{proof}

\subsection{Upgrading the discrepancy bound.}\
In this subsection, we upgrade \tref{fLL.over} using rigidity results for smooth linear statistics.  We will assume conditions stated in \tref{fLL.improved} throughout. We do not spend undue effort trying to optimize our bounds in $\beta$ or $V$ since our results are generally weaker than those of \cite{AS21}. Instead, our purpose is to show how our overcrowding estimates can be upgraded using known rigidity bounds for smooth linear statistics. In particular, we demonstrate that overcrowding bounds are sufficient to bound the absolute discrepancy, rather than just the positive part of the discrepancy. The mechanism for this is consists of finding a screening region of excess positive charge near the boundary of a ball with large absolute discrepancy. We note that the idea of obtaining a screening region is already present and features prominently in \cite{L21}, and we do not add anything fundamentally new to this procedure, but rather adapt it for our overcrowding estimate.

For $\alpha \in (0,1]$ and $R \in (0,\infty)$, let $\xi_{R,\alpha} : \R \to \R$ be a function satisfying
\begin{itemize}
\item $0 \leq \xi_{R,\alpha} \leq 1$
\item  $\xi_{R,\alpha}(x) = 1 \quad \forall x \in [-R,R]$ and $\xi_{R,\alpha}(x) = 0 \quad \forall x \not \in (-R-\alpha R, R+\alpha R)$
\item $\xi'_{R,\alpha}(x) \leq 0$ for $x \geq 0$. 
\item $\sup_{x \in \R} |\xi^{(k)}_{R,\alpha}(x)| \leq C_k (\alpha R)^{-k}$ for $k=1,2,3,4$.
\end{itemize}
In what follows, we will consider the map $x \mapsto \xi_{R,\alpha}(|x|)$ from $\R^d \to \R$, and by abuse of notation we will write this map as $\xi_{R,\alpha}$. We will also write $\Fluct(\phi) := \int \phi(x) \fluct(dx)$ for the ($N$-dependent) fluctuation measure $\fluct$ defined in \eref{fluctdef}.

\begin{theorem}[Corollary of \cite{S22}, Theorem 1] \label{t.serfaty}
Under the conditions on $W = V_N$ stated above \tref{fLL.improved}, for a large enough constant $C > 0$ and $\alpha R \geq C$, we have
\begin{equation} \label{e.serfaty}
	|\log \E^{V_N}_{N,\beta} \exp(t \Fluct(\xi_{R,\alpha}))| \leq \frac{C|t|R^{d-2}}{\alpha^3} + Ct^2 \(\frac{R^{d-4}}{\alpha^4} + \frac{R^{d-2}}{\alpha}\) + \frac{Ct^4 R^{d-8}}{\alpha^8} + \frac{C|t|R^d}{N^{2/d}}.
\end{equation}
for all $|t| \leq C^{-1}\alpha^2 R^2$. In particular for $t=1$ and $\alpha^3 R^2 \leq N^{2/d}$, we have
\begin{equation} \label{e.serfaty2}
	|\log \E^{V_N}_{N,\beta} \exp( \Fluct(\xi_{R,\alpha}))| \leq C\frac{|t|R^{d-2}}{\alpha^3}.
\end{equation}
\end{theorem}
\begin{proof}
A more general estimate on a non-blown up scale, i.e.\ the scale with interstitial distance of order $N^{-1/d}$, and with the thermal equilibrium measure $\mu_{\theta}$ in place of the equilibrium measure $\mueq$, is \cite[Theorem 1]{S22}. We slightly transform the estimate by using $\beta \geq C^{-1}$, changing the interstitial length scale to $O(1)$, and plugging in the specifics of $\xi_{R,\alpha}$. We also use (see \cite[Theorem 1]{AS19})
$$
\|\mu_{\theta} - \mueq \|_{L^\infty(\R^d)} \leq \frac{C}{N^{2/d}} 
$$
to replace the thermal equilibrium measure by $\mueq$, generating the $C|t|R^dN^{-2/d}$ term in \eref{serfaty}.
\end{proof}

The following proposition shows that one can find a screening region of excess positive charge whenever the absolute discrepancy is large in a ball. It is a natural consequence of rigidity for fluctuations of smooth linear statistics.
\begin{proposition} \label{p.disc.screen.region}
Fix $\delta \in (0,d)$ and $R \geq C$ for a large enough $C > 0$ and $\alpha \in (0,1]$. For any $T \geq 1$, if $\Fluct(\xi_{R-\alpha R,\alpha}) \leq \frac{T}{10}R^{\delta}$, we can find $k \in \Z^{\geq 0}$ and a constant $C_6 > 0$ such that
$$
\Disc(\Ann_{[R-\alpha_k R, R]}(z)) \geq \Disc(B_R(z)) - \frac{T}{2} R^\delta
$$
for $\alpha_k = \alpha - C_6^{-1} R^{-d + \delta} k$ with $C_6^{-1} R^{-d + \delta} k \leq \alpha_k \leq \alpha$. If instead $\Fluct(\xi_{R,\alpha}) \geq -\frac{T}{10}R^{\delta}$, we can also find $\alpha_k$ as above with
$$
\Disc(\Ann_{[R, R+\alpha_k R]}(z)) \geq -\Disc(B_R(z)) - \frac{T}{2} R^\delta.
$$
\end{proposition}
\begin{proof}
First, note that
$$
\Disc(B_{s+\ep}(z)) - \Disc(B_s(z))) = \int_{\Ann_{[s+\ep,s)}(z)} \fluct(dx)
$$
whenever $\fluct$ has no atoms on $\pa B_s(z)$ or $\pa B_{s+\ep}(z)$. Thus, we have by integration in spherical coordinates that
\begin{align*}
	\Fluct(\xi_{R-\alpha R,\alpha}) &= \int_0^{R} \xi_{R-\alpha R,\alpha}(s) d(\Disc(B_s(z))) = -\int_0^R \frac{d}{ds} \xi_{R-\alpha R,\alpha}(s) \Disc(B_s(z)) ds \\
	&= \Disc(B_{R}(z)) -\int_{R-\alpha R}^R \frac{d}{ds} \xi_{R-\alpha R,\alpha}(s) \(\Disc(B_s(z)) - \Disc(B_{R}(z))\) ds.
\end{align*}
Assuming now $\Fluct(\xi_{R-\alpha R,\alpha}) \leq \frac{T}{10} R^\delta$, by the mean value theorem, there must exist $s \in (R-\alpha R,R)$ such that $\Disc(B_s(z)) - \Disc(B_{R}(z)) \leq \frac{T}{10} R^{\delta} - \Disc(B_R(z)) $. This means that
$$
\Disc(\Ann_{[s,R)}(z)) \geq  \Disc(B_R(z)) - \frac{T}{10} R^{\delta}.
$$
For any $k$ such that $R - \alpha R \leq R - \alpha_k R  < s$, we have
\begin{align*}
\Disc(\Ann_{(R-\alpha_k R,R)}(z)) &\geq \Disc(\Ann_{[s,R)}(z)) - \mueq^N(\Ann_{(R-\alpha_k,s]}(z)) \\ &\geq \Disc(\Ann_{[s,R)}(z)) - CR^{d-1}|s - \alpha_k R|.
\end{align*}
We choose $k$ such that $|s-(R- \alpha_k R)| \leq C^{-1}R^{-d+1 + \delta}$ for a large enough constant $C$ to conclude.

If we instead assume $\Fluct(\xi_{R,\alpha}) \geq -\frac{T}{10} R^{\delta}$, then since
$$
\Fluct(\xi_{R,\alpha}) = \Disc(B_{R}(z)) -\int_{R}^{R+\alpha R} \frac{d}{ds} \xi_{R,\alpha}(s) \(\Disc(B_s(z)) - \Disc(B_{R}(z))\) ds,
$$
we can find $s \in (R,R+\alpha R)$ with
$$
\Disc(\Ann_{[R,s)}(z)) \geq -\frac{T}{10} R^{\delta} - \Disc(B_R(z)).
$$
Choosing $\alpha_k$ such that $s \leq R + \alpha_k R \leq R+\alpha R$ and $|s - (R+\alpha_k R)| \leq C^{-1} R^{-d+1+\delta}$ is sufficient to conclude.
\end{proof}

We are now ready to prove \tref{fLL.improved}.
\begin{proof}[Proof of \tref{fLL.improved}]
Let $\P = \P^{V_N}_{N,\beta}$. We choose $\alpha = R^{-\lambda}$ for $\lambda = 2/5$ and consider discrepancies in $B_R(z)$ of size $TR^{\delta}(1+\1_{d=2}\log R)$ for $\delta = d - 4/5$ and $T$ sufficiently large. We will however write a mostly generic argument in terms of $\lambda$ and $\delta$ and insert the specific values later. We will consider the case of positive discrepancy first.

If $\Fluct(\xi_{R-\alpha R,R}) \leq \frac{T}{10}R^{\delta}$, we can use \pref{disc.screen.region} to find $\alpha_k \in (0,\alpha]$ such that $\Disc(\Ann_{[R-\alpha_k R, R]}(z)) \geq \frac{T}{2}R^{\delta}(1+\1_{d=2}\log R)$. By a union bound, we have
\begin{align} \label{e.disc.breakdown}
	\lefteqn{ \P(\{\Disc(B_R(z)) \geq TR^{\delta}\log R\}) } \quad & \\ \notag &  \leq \P(\{\Fluct(\xi_{R-\alpha R,\alpha}) > \frac{T}{10}R^{\delta}\}) + C \alpha R^{d-\delta} \sup_{\alpha_k} \P( \{ \Disc(\Ann_{[R-\alpha_k R, R]}(z)) \geq \frac{T}{2}R^{\delta}\log R \}),
\end{align}
where the supremum is over all $\alpha_k = \alpha - C^{-1} R^{-d+\delta}k$, $k \in \Z$ and $\alpha_k \in [0,\alpha)$. Note that we have $\alpha_k \geq C^{-1} R^{-d+\delta} \gg R^{-1}$ always. Next, we bound the supremum in \eref{disc.breakdown}.

Let $\alpha' \in [C^{-1}R^{-d+\delta},\alpha]$ and let $\Omega$ be the $\alpha'$-thin annulus $\Ann_{[R-\alpha'R,R]}(z)$. We bound
$$
\P( \{ \Disc(\Omega) \geq \frac{T}{2}R^{\delta}(1+\1_{d=2}\log R )\} ) \leq \P( X(\Omega) \geq \mu(\Omega) + \rho |\Omega| )
$$
for $\rho = C^{-1} T (\alpha')^{-1} R^{-d+\delta} ( 1+\1_{d=2}\log R)$. Applying \pref{fLL.upbd} shows 
\begin{equation} \label{e.discfinal}
	\P(\{  \Disc(\Omega) \geq \frac{T}{2}R^{\delta} (1+\1_{d=2} \log R) \})  \leq e^{-c (\alpha')^{2/3}R^{2/3 + \delta} T} +  e^{-c (\alpha' R)^{d/3}  R^{2+2\delta - 2d} T^2} + e^{-(\alpha' R)^{d+2}}
\end{equation}
whenever $T (\alpha')^{-1/3}R^{\delta + 2/3 -d} \geq TR^{\delta + 2/3 -d  + \lambda/3}$ is large and we can estimate $$\| \mueq^N - q \|_{L^\infty(B_{2R}(z))} \leq C_2^{-1} TR^{\delta - d + \lambda}$$ for a constant $q$. The latter condition happens when $R^{1 + d - \delta - \lambda} \ll N^{1/d}$.

Considering the other term in \eref{disc.breakdown}, we apply \tref{serfaty} with $|t| = C^{-1}$ and Chebyshev's inequality to bound
\begin{equation} \label{e.fluct.cheby}
	\P(\{\Fluct(\xi_{R-\alpha R,\alpha}) > \frac{T}{10}R^{\delta}\}) \leq e^{C^{-1}(-\frac{T}{10} R^\delta +C R^{d-2+3\lambda})}.
\end{equation}
Note that $\alpha^3 R^2 = R^{4/5} \leq N^{2/d}$ so that \eref{serfaty2} applies. We thus apply our argument to parameters $\lambda$ and $\delta$ such that
\begin{equation}
	\delta + \frac23 - d + \frac{\lambda}{3} \geq 0, \quad \delta \geq d - 2 + 3\lambda.
\end{equation}
One can check that the smallest choice of $\delta$ is $\delta = d-4/5$ with $\lambda = 2/5$. With these choices, choosing $R \leq N^{\frac{5}{7d}}$ guarantees that we can approximate $\mueq^N$ by a constant sufficiently well.

Finally, we estimate the RHS of \eref{discfinal} and \eref{fluct.cheby} and plug them into \eref{disc.breakdown}. Note that $\alpha' \geq C^{-1}R^{-4/5}$, so the RHS of \eref{discfinal} is bounded by
$$
e^{-cR^{d-10/15}T} + e^{-cR^{2/5 + d/15}T^2} + e^{-R^{(d+2)/5}}.
$$
The factor $\alpha R^{d-\delta}$ within \eref{disc.breakdown} can be absorbed into the above at the cost of a constant factor within the exponent.

This finishes the one half of the proof of \tref{fLL.improved}. The proof of the lower bound on $\Disc(B_R(z))$ is nearly identical, except we use fluctuation bounds on $\xi_{R,\alpha}$ to find a screening region outside of $\pa B_R(z)$ with positive discrepancy.
\end{proof}

\addcontentsline{toc}{section}{References}
\bibliographystyle{alpha}
\bibliography{isotropicbib}{}

\vskip .5cm
\noindent
\textsc{Eric Thoma}\\
Courant Institute, New York University. \\
Email: {eric.thoma@cims.nyu.edu}.
\vspace{.2cm}

\end{document}